%% file: hamsim.tex
\documentclass[titlepage=off,DIV=9,abstracton]{scrartcl}
\usepackage[utf8]{inputenc}

\let\latextextsuperscript\textsuperscript
\AtBeginDocument{\let\textsuperscript\latextextsuperscript}

\usepackage{newtxtext,newtxmath}

\usepackage{microtype}
\usepackage{enumitem}
\usepackage{subcaption}

\usepackage{amsthm,amsmath,amsfonts,amssymb}
\usepackage{dsfont}

\usepackage[dvipsnames]{xcolor}
\definecolor{color3}{RGB}{140,29,49}
\definecolor{color2}{RGB}{31,107,76}
\definecolor{color1}{RGB}{29,66,91}
\usepackage[colorlinks=true,linkcolor=color1,urlcolor=color2,citecolor=color3,anchorcolor=green,pdfusetitle]{hyperref}

\usepackage[capitalize]{cleveref}
\Crefname{lemma}{Lemma}{Lemmas}
\Crefname{proposition}{Proposition}{Propositions}
\Crefname{definition}{Definition}{Definitions}
\Crefname{theorem}{Theorem}{Theorems}
\Crefname{conjecture}{Conjecture}{Conjectures}
\Crefname{corollary}{Corollary}{Corollaries}
\Crefname{example}{Example}{Examples}
\Crefname{section}{Section}{Sections}
\Crefname{appendix}{Appendix}{Appendices}
\Crefname{figure}{Fig.}{Figs.}
\Crefname{equation}{eq.}{eqs.}
\Crefname{table}{Table}{Tables}
\Crefname{item}{Property}{Properties}
\Crefname{remark}{Remark}{Remarks}

\newtheorem{theorem}{Theorem}
\newtheorem{definition}[theorem]{Definition}
\newtheorem{corollary}[theorem]{Corollary}

\newtheorem{lemma}[theorem]{Lemma}

\makeatletter
\newtheorem*{rep@theorem}{\rep@title}
\newcommand{\newreptheorem}[2]{%
	\newenvironment{rep#1}[1]{%
		\def\rep@title{#2 \ref{##1} (restated)}%
		\begin{rep@theorem}}%
		{\end{rep@theorem}}}
\makeatother
\newreptheorem{theorem}{Theorem}

\input{tiles}
\tikzstyle{qubit}=[circle,draw,fill,thick,
inner sep=0pt,minimum size=2mm]
\tikzstyle{mqubit}=[circle,draw,fill=white,thick,
inner sep=0pt,minimum size=2mm]
\tikzstyle{heavy}=[ultra thick]

\usepackage[backend=bibtex8,style=alphabetic,doi=false,isbn=false,url=false,maxbibnames=20,sorting=ynt,sortcites=true]{biblatex}
\renewbibmacro{in:}{}
\bibliography{bibliography}

\newbibmacro{string+doi}[1]{\iffieldundef{doi}{#1}{\href{https://dx.doi.org/\thefield{doi}}{#1}}}
\DeclareFieldFormat{title}{\usebibmacro{string+doi}{\mkbibemph{#1}}}
\DeclareFieldFormat[article]{title}{\usebibmacro{string+doi}{\mkbibquote{#1}}}
\DeclareFieldFormat[incollection]{title}{\usebibmacro{string+doi}{\mkbibquote{#1}}}

\input{braket}

\newcommand\ii{\mathrm i}

\newcommand{\1} {\ensuremath{\mathds 1}}

\newcommand{\field}[1] {\mathds{#1}}
\newcommand{\hs}{\mathcal H}
\newcommand{\cS}{\mathcal{S}}
\newcommand{\Hsim}{H\Sim}
\newcommand{\Heis}{^{\operatorname{Heis}}}
\DeclareMathOperator{\poly}{poly}

\newcommand\Sim{_\mathrm{sim}}
\newcommand\Tar{_\mathrm{target}}
\newcommand\Hist{_\mathrm{hist}}
\newcommand\Hor{^\mathrm{hor}}
\newcommand\Ver{^\mathrm{vert}}

\newcommand\Tri{_{\tikz{\fill[black] (0,0) -- (.16,0) -- (0,.16) -- cycle;}}}
\newcommand\Sq{_{\tikz{\fill[black] (0,0) -- (.16,0) -- (.16,.16) -- (0,.16) -- cycle;}}}
\newcommand\Flag{_{\tikz{\fill[blue] (0,0) circle[radius=.06]; \draw[line width=1pt] (0,0) -- (0,.16);}}}
\newcommand\Glag{_{\tikz{\fill[red] (0,0) circle[radius=.06]; \draw[line width=1pt] (0,0) -- (0,.16);}}}

\DeclareDocumentCommand{\norm}{m}{\ensuremath{\left\|#1\right\|}}

\DeclareMathOperator{\BigO}{O}

\newcommand{\R}{\mathds{R}}
\newcommand{\Z}{\mathds{Z}}

\title{\makebox[0pt]{Universal Translationally-Invariant Hamiltonians}\\[1cm]}

\author{
Stephen Piddock\textsuperscript{1,2,}\thanks{stephen.piddock@bristol.ac.uk}
\and
Johannes Bausch\textsuperscript{3,}\thanks{jkrb2@cam.ac.uk}
}
\date{
${}^1$School of Mathematics, University of Bristol, UK \\
${}^2$Heilbronn Institute for Mathematical Research, Bristol, UK \\
${}^3$DAMTP, University of Cambridge, UK
\\[1cm]
{January 21, 2020}
}

\begin{document}
\maketitle
\begin{abstract}
    In this work we extend the notion of universal quantum Hamiltonians to the setting of translationally-invariant systems.
    We present a construction that allows a two-dimensional spin lattice with nearest-neighbour interactions, open boundaries, and translational symmetry to simulate any local target Hamiltonian---i.e.\ to reproduce the whole of the target system within its low-energy subspace to arbitrarily-high precision.
    Since this implies the capability to simulate non-translationally-invariant many-body systems with translationally-invariant couplings, any effect such as characteristics commonly associated to systems with external disorder, e.g.\ many-body localization, can \emph{also} occur within the low-energy Hilbert space sector of translationally-invariant systems.
    Then we sketch a variant of the universal lattice construction optimized for simulating translationally-invariant target Hamiltonians.
    
    Finally we prove that qubit Hamiltonians consisting of Heisenberg or XY interactions of varying interaction strengths restricted to the edges of a connected translationally-invariant graph embedded in $\R^D$ are universal, and can efficiently simulate any geometrically local Hamiltonian in $\R^D$.
\end{abstract}

\clearpage
\tableofcontents

\section{Introduction}\label{sec:intro}
Quantum computing has shown promise as an emerging technology to solve instances of difficult computing tasks with higher efficiency than possible classically; from optimisation \cite{Gilyen2019,Montanaro2020}, linear algebra \cite{Harrow2009QuantumEquations,Berry2017}, search and number-theoretic problems \cite{Montanaro2016}, language processing \cite{Wiebe2019,Aaronson2018}, to machine learning \cite{Peruzzo2014,Wang2019,McClean_2016,Li2019SublinearQA,Bausch2019c}.
One particularly promising application in the realm of near-term quantum computing---without access to full error-correction---is that of analogue quantum simulation \cite{Lloyd1996,Knee2015,Bravyi2002,Babbush2018,Childs2019}.

An alternative to digital quantum simulation---where a physical system is simulated using methods akin to the ones we would employ on a classical computer---analogue simulation aims to reproduce the physics of a target Hamiltonian of interest with another Hamiltonian over which the user has full control, by implementing the target dynamics directly within the host system.
A rigorous definition of analogue simulation was given in \cite{CMP} (extending that in \cite{Bravyi-Hastings}), where intuitively the entirety of the target Hamiltonian is contained in the low energy space of the simulator Hamiltonian. This is strong enough to capture many important properties of the target systems, such as the full spectrum---eigenvectors and eigenvalues---, partition function, and correlation functions.

Understanding which quantum Hamiltonians can simulate each other is thus important not only for potentially implementing such simulations in the laboratory, but also to understand the behaviour, complexity, and phenomenology that can emerge within the simulating host system.

It was also shown in \cite{CMP} that simple families of Hamiltonians can in fact simulate all other finite dimensional Hamiltonians, a property called \emph{universality}.
For example systems of qubits consisting only of Heisenberg interactions (or alternatively XY interactions) with varying interactions strengths were both shown to be universal, even when restricted to the edges of a square lattice.
However, all universality results to date concerned families of interactions whose interaction strengths varied from site to site. 

Our first main result is the existence of a translationally-invariant universal simulator, which consists of the same interaction between nearest neighbours in a 2D grid.

\begin{theorem}
	\label{thm:2Duniversal}
	For a square lattice of qudits of size $W\times H$, there exist 1-local interactions $h^{(1)},h^{(2)}$, as well as 2-local interactions $h^{\mathrm{vert}(1)}$,$h^{\mathrm{vert}(2)}$,$h^{\mathrm{hor}(1)}$,$h^{\mathrm{hor}(2)}$ such that the family of Hamiltonians of the form 
	\[H\Sim(W,H):=\Delta_1 H_1 +\Delta_2 H_2 \]
	are universal;
	where $H_1$ and $H_2$ are both given by
	\[H_k=\sum_{i=1}^{W}\sum_{j=1}^{H}h^{(k)}_{(i,j)}+\sum_{i=1}^{W}\sum_{j=1}^{H-1}  h^{\mathrm{vert}(k)}_{(i,j),(i,j+1)}
	+\sum_{i=1}^{W-1}\sum_{j=1}^{H} h^{\mathrm{hor}(k)}_{(i,j),(i+1,j)}\]
	and $\Delta_1,\Delta_2 \in \R$ are scaling parameters, which depend on $W$ and $H$ as
	\begin{align*}
	    \Delta_2 = W\times H
	    \quad\text{and}\quad
	    \Delta_1 = W^5 \times H^5.
	\end{align*}
	
	To $(\Delta, \eta, \epsilon)$-simulate a specific target local Hamiltonian $H\Tar$ on $n$ qudits with interaction strengths at most $J_{\max}$, the lattice size $W\times H$ must be chosen to be
	\begin{align*}
	    W, H &=  \poly\left(\Delta, \frac{1}{\eta},\frac{1}{\epsilon},J_{\max}\right)2^{\poly(n)} .
	\end{align*}
\end{theorem}

\Cref{thm:2Duniversal} is an immediate consequence of \cref{lem:H_AB,lem:ham-sim-props}, wherein $H\Sim$ is constructed explicitly, using techniques from Hamiltonian complexity literature; more specifically by encoding computation into the spectrum of a translationally-invariant Hamiltonian.
Similar techniques were used to prove that the local Hamiltonian problem is $\mathrm{QMA_{EXP}}$-complete \cite{Gottesman2009} for a translationally-invariant spin chain in 1D; in our setting we utilize such a ``history state'' Hamiltonian to prepare a background field within which an already-known universal Hamiltonian emerges as effective low-energy theory.
We achieve this by combining the history state Hamiltonian with a tiling background, akin to the methods utilized in \cite{Bausch2017,Bausch2019b}.

History state constructions all have the property that the local dimension of the qudits at each site is constant, but potentially very large. This property directly translates to our result here.
We leave it as an open problem to work out how much this local dimension can be reduced as was done for the case of the local Hamiltonian problem in \cite{Bausch2016}.

Note that the Hamiltonian $H\Sim$ from \cref{thm:2Duniversal} does not have periodic boundary conditions (but open ones), and thus is not translationally invariant in the sense of commuting with a shift operator.
Although previous work on the local Hamiltonian problem shows that the hardness results can hold up even for Hamiltonians with periodic boundary conditions \cite[Ch.~5.8]{Gottesman2009}, we do not expect to be able to prove that such systems are universal.
The ground space in \citeauthor{Gottesman2009}'s hardness construction is at least $N$-fold degenerate due to the shift invariance.
In fact, demanding translational invariance implies that \emph{all} eigenspaces of a Hamiltonian are also symmetric with respect to this shift operator.
Consequently this would seem to make it impossible for a universality result to hold, given the importance of locality in the definition of simulation.

We also comment on the efficiency of this simulation: both the number of qudits $WH$ and the interaction strengths $\Delta_1,\Delta_2$ of the simulator system are exponential in the number of qudits of the target system. 
This violates the definition of \emph{efficient} simulation presented in \cite{CMP}, which requires a polynomial relation between these quantities for simulating local Hamiltonians; but we note that is unavoidable for a translationally-invariant system such as this.
Since the only free parameters left in the Hamiltonian $\Hsim$ are the side lengths $W$ and $H$, it is impossible to simulate all (exponentially many) local Hamiltonians on $n$ qudits, using a simulator Hamiltonians of size at most $\poly (n)$.

If we are only interested in simulating systems which are themselves translationally invariant, then this counting-based counterargument no longer applies. 
Indeed, in \Cref{sec:TItarget}, we sketch how to construct a translationally-invariant family of Hamiltonians that can simulate any translationally-invariant Hamiltonian in 2D on $n$ qudits, using at most $\poly(n)$ qudits in the simulator.

%%%%%%

In the final part of the paper, qubit interactions are considered.
However, since the methods of \cite{Gottesman2009,Bausch2016} do not apply, we are not able to prove universality for translationally-invariant systems.
Instead we consider qubit Hamiltonians of a single interaction type (Heisenberg interactions or XY interactions) restricted to the edges of a translationally-invariant graph, but with interaction strengths which are allowed to vary.

Let $G$ be a graph, and let $\cS$ be a set of 2-local interactions. We will use the notation $\cS_G$-Hamiltonians to denote the family of Hamiltonians of linear combinations of interactions from $\cS$ restricted to the edges of $G$. The interaction graph of any Hamiltonian in $\cS_G$-Hamiltonians must be a subgraph of $G$.
%We similarly define $\cS_G^+$-Hamiltonians to be the family of $\cS^+$-Hamiltonians whose interactions are restricted to the edges of $G$

It is already known that qubit interactions on a square lattice are universal; an argument that is based on using perturbative gadget constructions~\cite{Oliveira-Terhal}.
These constructions are only efficient (i.e.\ the interaction strengths are $O(\poly(n))$) when the target Hamiltonian is spatially sparse.
In \cref{sec:geometriclocality}, we generalise these results to $D$ spatial dimensions, and show that the square---or (hyper-)cubic---lattice is not a special case, but in fact any connected translationally-invariant graph will do.

\begin{theorem}
	\label{thm:TIsimGL}
	Let $G$ be a connected translationally-invariant graph embedded in $\R^D$ for $D\ge 2$.
	Let $\cS$ be $\{XX+YY+ZZ\}$ or $\{XX+YY\}$.
	Then $\cS_G$-Hamiltonians are universal and can efficiently simulate all geometrically local Hamiltonians in $\R^D$.
\end{theorem}
%Finally in Section~\ref{sec:TItarget}, we show that if we only want to simulate a translationally invariant family of Hamiltonians, then the overhead in the number of additional qudits required is only $\poly(n)$.

\subsection{Many-Body Localization in Translationally-Invariant Spin Models}
Many-body localization is an extensively-studied phenomenon occurring in many-body systems with disorder, instigated by \citeauthor{Anderson1958}'s seminal work on disorder-suppressed spin diffusion processes in lattices subject to a randomized energy potential \cite{Anderson1958}.
\citeauthor{Anderson1958} showed that at sufficiently low energy densities (as compared to the randomized potential), no transport takes place:
wavefunctions are concentrated within small spatial regions---they are \emph{localized}.

While there exists no strict definition of many-body localization, there exist various methods to verify properties attributed to the presence or absence of localization in experimental or theoretical models \cite{Smith2016,Smith2017a}, not all of which are equivalent or compatible \cite[Sec.~6.1.6]{crowley2017entanglement,Pancotti2018}.
For instance, one idea is to assess whether a system dynamically thermalizes, when an unbalanced initial state is supplied: if the state thermalizes, any memory of the initial imbalance is lost \cite{Schreiber2015}.
Another method is to measure the spread of correlations throughout the system:
the growth of entanglement entropy becomes an indicator of the diffusion velocity, and thus a measure of localization \cite{Enss2017} (cf.\ Lieb-Robinson bounds, \cite{Lieb1972}).
Since such systems exhibiting many-body localization can fail to thermalize at long time scales, they are candidate counterexamples for the eigenstate thermalization hypothesis \cite{Basko2006,DeTomasi2019}.

Despite many-body localization occurring naturally in systems with disorder, in recent years many translationally-invariant systems with similar indications have been studied, e.g.\ slow/nonergodic density relaxations in the thermodynamic limit \cite{Schiulaz2015}, or slow entanglement entropy growth, typical of localization in interacting systems with disorder, as mentioned; a phenomenon \citeauthor{Yao2016a} name \emph{quasi-many-body localization}  \cite{Yao2016a}.
In order to observe these quasi-localization effects, the disordered initial state method had been successfully employed \cite{Smith2017a}.
However, despite exhibiting localization-like effects, at \emph{very} long time scales ergodicity is restored in these systems;
only a tight binding model studied by \citeauthor{Vidal1998} prevents propagation of wave packets exactly \cite{Vidal1998}.

In \citeyear{Smith2017a}, \citeauthor{Smith2017a} empirically observed a translationally-invariant Hamiltonian on a one-dimensional lattice which exhibits dynamical localization without any external disorder \cite{Smith2017a}; the model the authors describe is that of spinless fermions, coupled via spin-$1/2$ particles; while the spin subsystem eventually equilibrates, the fermionic subsystem retains memory of the initial state \cite{Smith2017a}.
Advances in experimental and computational techniques have made it possible to study many more such effects, and we refer the reader to the extensive overview and discussion presented in \cite{Pancotti2019}, in which the author provide evidence for non-thermalizing behaviour in the quantum East model---where the Hilbert space fractions into small blocks, protected from transitions between the different sectors \cite{Rakovszky2019}.

As discussed in \cite[Sec.~6.3]{crowley2017entanglement}, a translationally-invariant many-body system is ``intrinsically incapable of reproducing'' all of the effects associated to many-body localization.

In contrast, in this work we show that within the framework of quantum simulation, a translationally-invariant system can---within its low-energy subspace---simulate any non-translationally-invariant Hamiltonian; in particular, this implies that the entire spectrum, and thus all of the observable effects such as correlations and thermalization behaviour, of a disordered local Hamiltonian can be reproduced to arbitrary precision by translationally-invariant couplings within the simulator system.

As mentioned in \cref{sec:intro}---due to the definition of simulation, which only requires to reproduce the spectrum of a target Hamiltonian within the low-energy subspace of the simulator, and due to a parameter counting argument---there exists an exponential spatial overhead when simulating a translationally non-invariant many-body system with a translationally-invariant one.
In essence, this means that any such many-body localization effects have to be concentrated within a small lattice region, and thus also only a small region of the Hilbert space.
Yet if our aim is to reproduce random interaction strengths across the entire spin lattice---instead of reproducing a \emph{specific} set of random couplings---similar methods as the ones developed herein can be used to achieve this effect (see the discussion in \Cref{sec:conclusion}).

\section{Background and Preliminaries}
\subsection{Universal Hamiltonians}
Here we provide the formal definitions of the simulation concepts introduced in the introduction.
\begin{definition}[Special case of definition in~\cite{CMP}; variant of definition in~\cite{Bravyi-Hastings}] 
	\label{def:sim}
	We say that $H'$ is a $(\Delta,\eta,\epsilon)$-simulation of $H$ if there exists a local isometry $V = \bigotimes_i V_i$ such that:
	\begin{enumerate}
		\item 
		There exists an isometry $\widetilde{V}$ such that $\widetilde{V} \widetilde{V}^\dag = S_{\le \Delta(H')}$ and $\|\widetilde{V} - V\| \le \eta$;
		\item 
		$\| H'_{\le \Delta} - \widetilde{V}H\widetilde{V}^\dag \| \le \epsilon$.
	\end{enumerate}
	We say that a family $\mathcal{F}'$ of Hamiltonians can simulate a family $\mathcal{F}$ of Hamiltonians if, for any $H \in \mathcal{F}$ and any $\eta,\epsilon >0$ and $\Delta \ge \Delta_0$ (for some $\Delta_0 > 0$), there exists $H' \in \mathcal{F}'$ such that $H'$ is a $(\Delta,\eta,\epsilon)$-simulation of $H$.
	We say that the simulation is efficient if, in addition, for $H$ acting on $n$ qudits and $H'$ acting on $m$ qudits , $\|H'\| = \poly(n,1/\eta,1/\epsilon,\Delta)$ and $m = \poly(n,1/\eta,1/\epsilon,\Delta)$; $H'$ is efficiently computable given $H$, $\Delta$, $\eta$ and $\epsilon$; and each isometry $V_i$ maps to $O(1)$ qudits.
\end{definition}

\begin{definition}[\cite{CMP}]
	\label{def:universal}
	We say that a family of Hamiltonians is \emph{universal} if \emph{any} (finite-dimensional) Hamiltonian can be simulated by a Hamiltonian from the family.
	We say that the universal simulator is \emph{efficient} if the simulation is efficient for all local Hamiltonians.
\end{definition}

\begin{theorem}
	\label{thm:classification}
	Let $\mathcal{S}$ be any fixed set of two-qubit and one-qubit interactions such that $\mathcal{S}$ contains at least one interaction which is not 1-local. Then:
	\begin{itemize}
		\item If there exists $U \in SU(2)$ such that $U$ locally diagonalises $\mathcal{S}$, then $\mathcal{S}$-Hamiltonians are universal classical Hamiltonian simulators~\cite{Science};
		\item Otherwise, if there exists $U \in SU(2)$ such that, for each 2-qubit matrix $H_i \in \mathcal{S}$, $U^{\otimes 2} H_i (U^\dagger)^{\otimes 2} = \alpha_i Z^{\otimes 2} + A_i \otimes \1 + \1 \otimes B_i$, where $\alpha_i \in \R$ and $A_i$, $B_i$ are arbitrary single-qubit Hamiltonians, then $\mathcal{S}$-Hamiltonians are universal stoquastic Hamiltonian simulators~\cite{Bravyi-Hastings,Cubitt-Montanaro};
		\item Otherwise, $\mathcal{S}$-Hamiltonians are universal quantum Hamiltonian simulators.
	\end{itemize}
\end{theorem}

The following theorem, proven in \cite{CMP}, describes the simulation ability of XY and Heisenberg interactions when restricted to the edges of a 2D square lattice. We will generalise this result to higher spatial dimensions and other interaction patterns in \Cref{sec:qubits}.
\begin{theorem}[\cite{CMP}]
    \label{thm:XYsquare}
    Let $\square$ denote a 2D square lattice, and let $\cS$ be $\{XX+YY+ZZ\}$ or $\{XX+YY\}$. 
    Then $\cS_{\square}$-Hamiltonians are universal, and can efficiently simulate all spatially sparse local Hamiltonians
\end{theorem}

\subsection{Perturbation Theory}

In order to prove our main result we will need the following first order perturbation theory result about simulation, from \cite{Bravyi-Hastings} and rephrased slightly for our definition of simulation in \cite{CMP}.

\begin{lemma}[First-order simulation~\cite{Bravyi-Hastings, CMP}]
  \label{lem:firstorder}
  Let $H_0$ and $H_1$ be Hamiltonians acting on the same space, and let $\Pi_-$ be the projector onto the ground space of $H_0$. Suppose that $H_0$ has spectral gap $\ge 1$ and that there exists a local isometry $V$ such that $VV^{\dagger}=\Pi_-$ and
  \begin{equation}
\norm{V H\Tar V^\dag - \Pi_-H_1 \Pi_-} \le \epsilon/2.
  \end{equation}
  Then $H\Sim = \Delta H_0 + H_1$ $(\Delta/2,\eta,\epsilon)$-simulates $H\Tar$, provided that $\Delta \ge O(\|H_1\|^2/\epsilon + \|H_1\| / \eta)$.
\end{lemma}

\subsubsection{Mediator Gadgets}
\label{sec:PMmediator}
Oliveira and Terhal developed perturbative gadgets which "mediate" a 2-local interaction between qubits which do not directly interact, but which interact with a common "mediator" qubit. These gadgets were called subdivision, fork and crossing gadgets for the three slightly different ways in which they can be applied.

In \cite{Piddock-Montanaro}, a similar set of gadgets were constructed which use only XY interactions or Heisenberg interactions.
The effect of these interactions is described in the language of simulations in \Cref{lem:PMgadgets}, and in \Cref{fig:PMsubdivision} and \Cref{fig:PMgadgets}.
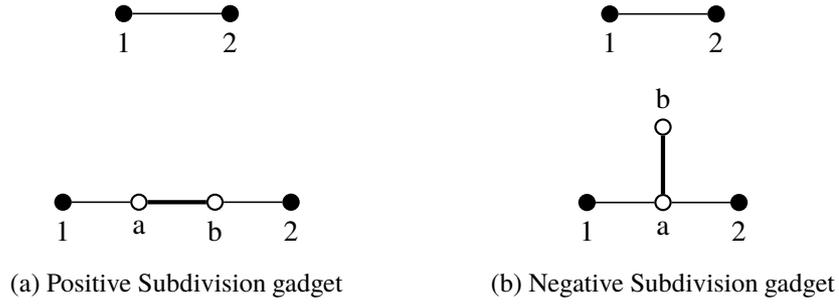
\begin{figure}[t]
	\centering
	\begin{subfigure}[b]{0.5\textwidth}
		\centering
		\begin{tikzpicture}
		\begin{scope}[shift={(-0.5,0)}]
		\node[qubit] (1) at (-1,1)[label=below:1]{};
		\node[qubit] (2) at (2,1) [label=below:2] {};
		\node[mqubit] (a) at (0,1)[label=below:a]{};
		\node[mqubit] (b) at (1,1)[label=below:b]{};
		
		\draw [black, semithick] (1) to (a);
		\draw [heavy] (a) to (b);
		\draw [black, semithick] (2) to (b);
		\end{scope}
		
		\begin{scope}[shift={(0,2.5)}]
		\node[qubit] (1) at (-0.7,1)[label=below:1]{};
		\node[qubit] (2) at (0.7,1) [label=below:2] {};
		
		\draw [black, semithick] (1) to (2);
		\end{scope}
		\end{tikzpicture}
		\caption{Positive Subdivision gadget}
		\label{fig:subdivisionpos}
	\end{subfigure}
	\begin{subfigure}[b]{0.4\textwidth}
		\centering
		\begin{tikzpicture}
		\begin{scope}
		\node[qubit] (1) at (-1,1)[label=below:1]{};
		\node[qubit] (2) at (1,1) [label=below:2] {};
		\node[mqubit] (a) at (0,1)[label=below:a]{};
		\node[mqubit] (b) at (0,2)[label=above:b]{};
		
		\draw [black, semithick] (1) to (a);
		\draw [heavy] (a) to (b);
		\draw [black, semithick] (2) to (a);
		\end{scope}
		
		\begin{scope}[shift={(0,2.5)}]
		\node[qubit] (1) at (-0.7,1)[label=below:1]{};
		\node[qubit] (2) at (0.7,1) [label=below:2] {};
		
		\draw [black, semithick] (1) to (2);
		\end{scope}
		\end{tikzpicture}
		\caption{Negative Subdivision gadget}
		\label{fig:subdivisionneg}
	\end{subfigure}
	\caption{Subdivision gadgets from~\cite{Piddock-Montanaro} for Heisenberg and XY interactions.
	In each case the top interaction pattern is simulated using the gadget underneath.
	White vertices denote mediator qubits and the thick lines indicate a heavy 2-local interaction term. }
	\label{fig:PMsubdivision}
\end{figure}

\begin{lemma}[\cite{Piddock-Montanaro}]
	\label{lem:PMgadgets}
	Let $H_0= h_{ab}$ and $\lambda, \mu \in \R$. 
	Let $\epsilon, \eta >0$ and let $\Delta \gg \frac{1}{\epsilon^2} +\frac{1}{\eta^2}$.
	Then $\Delta H_0 +\Delta^{\frac{1}{2}} H_2 +H_1$ $(\Delta/2, \epsilon,\eta)$-simulates the following interactions:
	\begin{enumerate}
		\item Subdivision \begin{enumerate}
			\item $\lambda h_{12}$ \quad when $H_2=h_{1a} +\lambda h_{2b}$ and $H_1=0$.
			\item $-\lambda h_{12}$ \quad when $H_2=h_{1a} +\lambda h_{2a}$ and $H_1=0$.
		\end{enumerate}
		\item Fork $\lambda h_{13} + \mu h_{23}$ \quad when $H_2=h_{3b} -\lambda h_{1a}-\mu h_{2a}$ and $H_1=\lambda \mu h_{12}$.
		\item Crossing $\lambda h_{14} + \mu h_{23} \quad \text{ when} \begin{array}{l}
		H_2=h_{1a}+h_{2a}-\lambda h_{b4} -\mu h_{b4} \\ H_1= h_{12}-\lambda h_{24} -\mu h_{13}+\lambda\mu h_{34}. \end{array}$
		
	\end{enumerate}
	See Figure~\ref{fig:PMsubdivision} and Figure~\ref{fig:PMgadgets} which shows the interaction graphs for each of these gadgets.
\end{lemma}

Furthermore, these gadgets can be applied to multiple qubits in a system in parallel without separate gadgets interfering with each other; see \cite{CMP} for details.

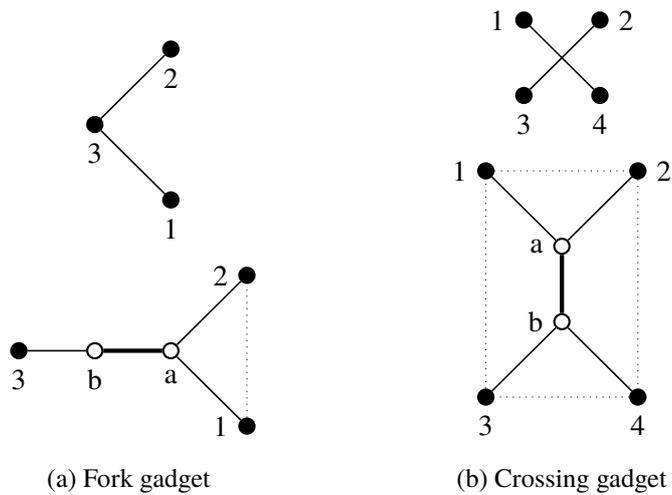
\begin{figure}[t]
	\centering
	\begin{subfigure}[b]{0.4\textwidth}
	\centering
		\begin{tikzpicture}
		\begin{scope}[rotate=90,yscale=-1,xscale=1]
		\node[qubit] (1) at (-1,2) [label=left:1] {};
		\node[qubit] (2) at (1,2)[label=left:2]{};
		\node[mqubit] (a) at (0,1)[label=below:a]{};
		\node[mqubit] (b) at (0,0)[label=below:b]{};
		\node[qubit] (3) at (0,-1)[label=below:3]{};
		
		\draw [heavy] (a) to (b);
		\draw [black, semithick] (1) to (a);
		\draw [black, semithick] (2) to (a);
		\draw [black, semithick] (3) to (b);
		\draw [dotted] (1) to (2);
		\end{scope}
		
		\begin{scope}[shift={(0,3)},rotate=90,yscale=-1,xscale=1]
		\node[qubit] (1) at (-1,1) [label=below:1] {};
		\node[qubit] (2) at (1,1)[label=below:2]{};
		\node[qubit] (3) at (0,0)[label=below:3]{};

		\draw [black, semithick] (3) to (1);
		\draw [black, semithick] (3) to (2);
		
		\end{scope}
		\end{tikzpicture}
		\caption{Fork gadget}
		\label{fig:fork}
	\end{subfigure}
	\begin{subfigure}[b]{0.4\textwidth}
		\centering
		\begin{tikzpicture}
		\begin{scope}
		\node[qubit] (1) at (-1,2) [label=left:1] {};
		\node[qubit] (2) at (1,2)[label=right:2]{};
		\node[mqubit] (a) at (0,1)[label=left:a] {};
		\node[mqubit] (b) at (0,0)[label=left:b]{};
		\node[qubit] (3) at (-1,-1)[label=below:3]{};
		\node[qubit] (4) at (1,-1)[label=below:4]{};
		
		\draw [heavy] (a) to (b);
		\draw [black, semithick] (1) to (a);
		\draw [black, semithick] (2) to (a);
		\draw [black, semithick] (3) to (b);
		\draw [black, semithick] (4) to (b);
		\draw [dotted] (1) to (2);
		\draw [dotted] (2) to (4);
		\draw [dotted] (4) to (3);
		\draw [dotted] (3) to (1);
		\end{scope}
		
		\begin{scope}[shift={(0,3)}]
		\node[qubit] (1) at (-0.5,1) [label=left:1] {};
		\node[qubit] (2) at (0.5,1)[label=right:2]{};
		\node[qubit] (3) at (-0.5,0)[label=below:3]{};
		\node[qubit] (4) at (0.5,0)[label=below:4]{};
		\draw [black, semithick] (1) to (4);
		\draw [black, semithick] (2) to (3);
		\end{scope}
		\end{tikzpicture}
		\caption{Crossing gadget}
		\label{fig:crossing}
	\end{subfigure}
	
	\caption[Subdivision, fork and crossing gadgets.]{%
		Fork and crossing gadgets from~\cite{Piddock-Montanaro} for Heisenberg or XY interactions.
		In each case the top interaction pattern is simulated using the gadget underneath.
		White vertices denote mediator qubits and the thick lines indicate a heavy 2-local interaction term.}
	\label{fig:PMgadgets}
\end{figure}

\clearpage
\section{Translationally-Invariant Qudit Hamiltonians}\label{sec:construction}

In this section, we construct a simulator Hamiltonian $H\Sim$ with translationally-invariant nearest-neighbour interactions on a square lattice with open boundary conditions, and such that its low energy space approximates an arbitrary target Hamiltonian $H\Tar$.
Since the only parameter to encode the target Hamiltonian's interactions in $H\Sim$ is the size of the simulating lattice $\Lambda\Sim$, there will naturally be an exponential overhead in the number of lattice sites as compared to the original Hilbert space of the target Hamiltonian, and so we are not concerned with the simulation being efficient here.
For the special case where the target Hamiltonian is itself translationally invariant, we present a modified variant in \Cref{sec:TItarget}.

To prove our main result, \Cref{thm:2Duniversal}, we recall that the Heisenberg interaction ($h\Heis=XX+YY+ZZ$) was proven to be universal in \cite{CMP}, even when restricted to the edges of a square lattice, see \Cref{thm:XYsquare}.
It will therefore suffice to simulate the following target Hamiltonian
\[H\Tar=\sum_{i=1}^{n}\sum_{j=1}^{n-1} \alpha_{ij} h\Heis_{(i,j),(i,j+1)}+\sum_{i=1}^{n-1}\sum_{j=1}^{n} \beta_{ij} h\Heis_{(i,j),(i+1,j)}\]
for arbitrary choices of $\alpha_{ij}$ and $\beta_{ij}$.

We will first prove this result assuming there exists a $H_A$ with some nice ground state properties, as outlined in \Cref{thm:HAexistence}.
Then in \Cref{sec:construction}, we will show that such a Hamiltonian can indeed be constructed.

%TODO maybe mention QMAEXP vs QMA constructions

First we show that it suffices to construct a translationally-invariant 2D lattice Hamiltonian with a constant spectral gap whose groundstate encodes the interaction strengths of the target Hamiltonian.

\begin{theorem}
	\label{thm:HAexistence}
	There exists a  translationally-invariant 2-local Hamiltonian $H_A$ on a 2D grid of qudits of dimension $d$ for some constant $d$, and a growing function $\Delta_2(W,H)$ such that for any choice of real numbers $\{\alpha_{ij}\}_{i,j \in \{1,\dots n\}}$ and $\{\beta_{ij}\}_{i,j \in \{1,\dots n\}}$ and $\delta>0$, there exists a choice of side lengths $W$ and $H$ for $H_A$ such that the following is true:
	\begin{enumerate}
		\item There is a unique ground state $\ket{\Psi_0}$ and the spectral gap is $\Omega(1)$
		\item There exist projectors $P^{(1)},P^{(2)}$ and $P^{(3)}$ on  $\field{C}^d$ such that\[ \bra{\Psi_0}P^{(3)}_{(i,j)}\ket{\Psi_0}= \begin{cases}
		0 & \text{ if }i\le n \text{ and } j\le n\\
		1 & \text{ otherwise}\\
		\end{cases} \]
		and for $i,j \le n$ we have
		\[\left|\bra{\Psi_0}P^{(1)}_{(i,j)}\ket{\Psi_0}-\frac{\alpha_{ij}}{\Delta_2}\right| \le \delta  \qquad \text{ and } \qquad \left| \bra{\Psi_0}P^{(2)}_{(i,j)}\ket{\Psi_0}-\frac{\beta_{ij}}{\Delta_2}\right| \le \delta.
		\]
		
	\end{enumerate}
\end{theorem}
We remark that although the function $\Delta_2$ may look like an unusual part of the requirements of the Hamiltonian $H_A$, it is necessary because otherwise it would be impossible for $\bra{\Psi_0}P^{(1)}_{(i,j)}\ket{\Psi_0}$ to approximate any $\alpha_{ij},\beta_{ij}>1$.

We now prove that the existence of a Hamiltonian satisfying the properties of \Cref{thm:HAexistence} allows us to construct a 2D translationally-invariant Hamiltonian $H\Sim$ which can simulate an arbitrary (not necessarily translationally-invariant) Hamiltonian of Heisenberg interactions on a square lattice.
\begin{lemma}
	\label{lem:H_AB}
	Let $\mathcal{H}_A=(\field{C}^d)^{\otimes WH}$ and $\mathcal{H}_B=(\field{C}^2)^{\otimes WH}$ be two $W\times H$ sized grids of qudits and qubits respectively.
	Let $H_A$ acting on $\mathcal{H}_A$ satisfy the conditions of \Cref{thm:HAexistence}. 
	Let $P^{(1)},P^{(2)}$ and $P^{(3)}$ be projectors on $\field{C}^d$, and let
	
	\[H_{AB}= \sum_{i=1}^{W}\sum_{j=1}^{H-1} P^{(1)}_{(i,j)} \otimes h\Heis_{(i,j),(i,j+1)}+\sum_{i=1}^{W-1}\sum_{j=1}^{H} P^{(2)}_{(i,j)} \otimes h\Heis_{(i,j),(i+1,j)}
	\]
	act on $\mathcal{H}_A \otimes \mathcal{H}_B$.
	
	Then, provided that $\Delta_1 \ge \left(\Delta_2^2 W^2H^2/\epsilon +\Delta_2 WH/\eta \right)$, and $\delta \le \epsilon/(4\Delta_2 n^2)$, the Hamiltonian
	\[H\Sim:=\Delta_1 \left(H_A \otimes I_B +\sum_{i=1}^W\sum_{j=1}^H P^{(3)}_{(i,j)} \otimes \proj{1}_{(i,j)}\right) +\Delta_2 H_{AB} \]
	$(\Delta_1/2, \eta,\epsilon)$-simulates the Hamiltonian $H\Tar$
	\[H\Tar=\sum_{i=1}^{n}\sum_{j=1}^{n-1} \alpha_{ij} h\Heis_{(i,j),(i,j+1)}+\sum_{i=1}^{n-1}\sum_{j=1}^{n} \beta_{ij} h\Heis_{(i,j),(i+1,j)}.\]
	
\end{lemma}

\begin{proof}
	
	We prove this result by applying Lemma~\ref{lem:firstorder}. To match the notation in Lemma~\ref{lem:firstorder}, let $H_0=H_A \otimes I_B +\sum_{i=1}^W\sum_{j=1}^H P^{(3)}_{(i,j)} \otimes \proj{1}_{(i,j)}$ and $H_1=\Delta_2 H_{AB}$. 
	We decompose $\mathcal{H}_B=(\field{C}^2)^{\otimes WH}=(\field{C}^2)^{\otimes n^2}\otimes (\field{C}^2)^{\otimes (WH-n^2)}=\mathcal{H}_{B_1}\otimes \mathcal{H}_{B_2}$, where $\mathcal{H}_{B_1}$ contains the $n^2$ qubits in the corner of the grid where $\bra{\Psi_0}P^{(3)}_{(i,j)}\ket{\Psi_0}= 0$.  
	Note that with respect to this split, the groundspace of $H_0$ is equal to $\ket{\Psi_0}\otimes(\field{C}^2)^{\otimes n^2} \otimes \ket{0}^{(WH-n^2)}$.
	
	We take $V: \ket{\phi} \rightarrow \ket{\Psi_0} \otimes \ket{\phi}\otimes \ket{0}^{\otimes(WH-n^2)}$ to be the local isometry from $(\field{C}^2)^{n^2}$ to $\mathcal{H}_A \otimes \mathcal{H}_B$, noting that $V V^{\dagger}= \proj{\Psi_0}\otimes I_{B_1}\otimes \proj{0}^{\otimes(WH-n^2)}$ is the projector onto the ground space of $H_0$. 
	It is straightforward to check that 
	\begin{align}\Pi_- H_1 \Pi_-&=\left(\proj{\Psi_0}\otimes I_{B_1}\otimes \proj{0}^{\otimes(WH-n^2)}\right)\Delta_2 H_{AB} \left(\proj{\Psi_0}\otimes I_{B_1}\otimes \proj{0}^{\otimes(WH-n^2)}\right)\\
	&= \proj{\Psi_0}\otimes \left(\begin{array}{l}
	\sum_{i=1}^{n}\sum_{j=1}^{n-1} \Delta_2 \bra{\Psi_0}P^{(1)}_{(i,j)}\ket{\Psi_0} h\Heis_{(i,j),(i,j+1)}\\[2mm]
	\quad + \sum_{i=1}^{n-1}\sum_{j=1}^{n}  \Delta_2 \bra{\Psi_0}P^{(3)}_{(i,j)}\ket{\Psi_0} h\Heis_{(i,j),(i+1,j)}
	\end{array} \right) \otimes \proj{0}^{\otimes(WH-n^2)}\\
	&=V \tilde{H}\Tar V^{\dagger}
	\end{align}
	where 
	\[\tilde{H}\Tar = \sum_{i=1}^{n}\sum_{j=1}^{n-1} \Delta_2 \bra{\Psi_0}P^{(1)}_{(i,j)}\ket{\Psi_0} h\Heis_{(i,j),(i,j+1)}+\sum_{i=1}^{n-1}\sum_{j=1}^{n} \Delta_2 \bra{\Psi_0}P^{(3)}_{(i,j)}\ket{\Psi_0} h\Heis_{(i,j),(i+1,j)}\]
	Finally, recalling that by Lemma~\ref{thm:HAexistence}   \[\left|\bra{\Psi_0}P^{(1)}_{(i,j)}\ket{\Psi_0}-\frac{\alpha_{ij}}{\Delta_2}\right| \le \delta  \qquad \text{ and } \qquad \left| \bra{\Psi_0}P^{(2)}_{(i,j)}\ket{\Psi_0}-\frac{\beta_{ij}}{\Delta_2}\right| \le \delta,  \]
	%\[\bra{\Psi_0}P^{(1)}_{(i,j)}\ket{\Psi_0}=\alpha_{ij}/\Delta_2 \qquad \text{ and } \qquad \bra{\Psi_0}P^{(2)}_{(i,j)}\ket{\Psi_0}=\beta_{ij}/\Delta_2 \] 
	and so 
	\begin{align*}
	\norm{\tilde{H}\Tar-H\Tar} &\le \sum_{i=1}^{n}\sum_{j=1}^{n-1} \Delta_2 \delta \norm{  h\Heis_{(i,j),(i,j+1)}}+\sum_{i=1}^{n-1}\sum_{j=1}^{n}\Delta_2 \delta \norm{  h\Heis_{(i,j),(i+1,j)}} \\
	&= 2n(n-1)\Delta_2 \delta \le \epsilon/2.
	\end{align*}
	The result now follows from Lemma~\ref{lem:firstorder} and  noting that $\norm{H_{AB}}  \le  WH$.
\end{proof}

%Since we know that for instance Heisenberg interactions on a square lattice are universal, we just need to construct a translationally invariant Hamiltonian $H_A$ with ground state $\ket{\Psi_0}$ such that there exists a projector $\Pi_i$ acting on each subsystem of $\mathcal{H}_A$ such that we can control $\bra{\Psi_0} \Pi_i \ket{\Psi_0}$ for all $i$.

\label{sec:TIuniversality}

\subsection{Simulator Hamiltonian Construction}\label{sec:construction-overview}
In this section we constructively prove the existence of a Hamiltonian described in \Cref{thm:HAexistence}.
The Hamiltonian $H_A$ therein will be defined on a square qudit lattice $\Lambda\Sim$, with one- and two-local horizontal and vertical nearest neighbour interactions; these interactions will be translationally invariant.

The overall construction will be as follows.
\begin{enumerate}
	\item We construct a tile set that translates the lattice side lengths $W \times H$ to binary numbers on the left and lower edge, respectively; $W=2^n + 2^b$ for some $n,b\in\field N$, $n<b/2$, and $H\in\field N$, $b \le H$ arbitrary.
	\item We use the single $1$ located at position $b$ offset from the lower left corner on the left edge to mark a right triangle of side length $b$ in the lower left of $\Lambda\Sim$.
	\item We use the single $1$ located at position $n$ offset from the lower left corner on the left edge to mark a square of size $n\times n$ in the lower left of $\Lambda\Sim$, which sits atop the triangle.
	\item Define a history state Hamiltonian that executes a quantum Turing machine with a 1D tape, which is defined such that it starts in the lower left corner of $\Lambda\Sim$, continues horizontally until it encounters the diagonal boundary of the triangle, where it wraps around towards the left edge. In the same fashion it winds upwards across the entire marked triangle.
	\item The program that the history state Hamiltonian executes is such that at each location $(i,j)$ within the marked square a qubit $\ket{\theta_{i,j}}$ is rotated into a superposition
	\begin{equation}\label{eq:theta-state}
	\ket{\theta_{i,j}} := \cos{\theta_{i,j}} \ket 0 + \sin{\theta_{i,j}} \ket 1,
	\end{equation}
	where the $\theta_{i,j}$ are specified (to some given precision) within the binary string on the horizontal lower edge of the lattice, i.e.\ the one encoding the number $H$; this is repeated for two families of real numbers necessary to satisfy \Cref{lem:H_AB}.
\end{enumerate}

Of course, since we do not have an arbitrary gate set available, the Turing machine will have to approximate the single qubit rotations, yielding only approximate angles $\ket{\tilde\theta_{i,j}} \approx \ket{\theta_{i,j}}$.
This is achieved by a standard application of Solovay-Kitaev; where we note that the target rotation can be approximated to high precision, as the circuit depth of an SK-compiled single-qubit gate is $\BigO(\log^4 \epsilon^{-1})$ in the required target precision $\epsilon$.

We will formalize this construction in the following technical lemma; a corollary of which will be that the resulting Hamiltonian immediately satisfies \Cref{thm:HAexistence}.
\begin{lemma}\label{lem:ham-sim-props}
	Let $n\in\field N$, and a polynomial $p$ such that $\delta(n) = 1/\exp p(n)$.
	Take two families of real numbers $\{ \alpha_{ij} \}_{i,j\in\{1,\ldots,n\}}$, $\{ \beta_{ij} \}_{i,j\in\{1,\ldots,n\}}$, and let $\epsilon>0$.
	Each $\alpha_{ij},\beta_{ij}\in [0,\Delta_2]$, as well as $\Delta_2$ and $\epsilon$, have bit precision $\Xi\in\field N$.
	Then the following exists.
	A spin lattice $\Lambda\Sim$ of side length $W\times H$, and interactions $h^{(1)},h\Hor,h\Ver$, such that
	\begin{enumerate}
		\item $H=\exp\poly(n,\Xi)$.
		\item $W=2^n+2^b$ for some $b\in\field N$, $n<b/2$, $b \ge \lceil \log_2 H \rceil$, $b=\poly(n,\Xi)$.
		\item $h^{(1)},h\Hor,h\Ver$ are independent of $W$ and $H$.
		\item All terms have matrix entries $\in S:=\{\pm 1,\pm\sqrt 2,  \pm\sqrt{\ii}\}$.
	\end{enumerate}
	Define
	\[
	H_A:=\sum_{i\in\Lambda\Sim} h^{(1)}_{(i)} + \sum_{i=1}^{W-1} \sum_{j=1}^H h\Hor_{(i,j),(i+1,j)} + \sum_{i=1}^W \sum_{j=1}^{H-1} h\Ver_{(i,j),(i,j+1)},
	\]
	and let $\ket{\Psi_i}$ be the eigenvectors of $H_A$, with spectrum $\lambda_0<\lambda_1<\ldots$.
	Denote with $\Pi\Tri,\Pi\Sq,\Pi\Flag,\Pi\Glag$ four projectors onto four distinct local basis states.
	Then
	\begin{enumerate}
		\item $\Delta(H_A) := \lambda_1-\lambda_0 \ge 1/\poly b.$
		\item All eigenstates $\ket{\Psi_i}$ are product across a partition of the lattice $\Lambda\Sim=\Lambda\Tri \times \Lambda\Tri^c$, where $\Lambda\Tri$ is a right triangular cut of the lattice with non-hypotenuse side length $b$ in the lower left of $\Lambda\Sim$, and $\Lambda\Tri^c$ its complement.
		\item Across this partition, all $\ket{\Psi_i} = \ket{\Phi_i} \otimes \ket{R_i}$ where the latter is a product state.
		\item $\Pi\Tri\ket{R_0}=0$, and $\Pi\Tri\ket{\Phi_0} = \ket{\Phi_0}$.
		\item $\ket{\Phi_0}=\ket{\phi\Hist} \otimes \ket{\phi\Tri}$, such that the latter is a product state that factors like $\ket{\phi\Tri} = \ket{\phi\Sq} \otimes\ket{\phi\Sq^c}$, where $\ket{\phi\Sq}$ is defined on a square of side length $n\times n$ in the lower left of $\Lambda\Tri$ (denoted $\Lambda\Sq$), and $\ket{\phi\Sq^c}$ on its complement; we have $\Pi\Sq\ket{\phi\Sq^c}=0$, and $\Pi\Sq\ket{\phi\Sq} = \ket{\phi\Sq}$.
		\item for all $(i,j)\in\Lambda\Sq$, we have
		\[
		\bra{\Psi_0} \Pi\Flag^{(i,j)} \ket{\Psi_0} = \tilde\alpha_{ij}
		\quad\text{and}\quad
		\bra{\Psi_0} \Pi\Glag^{(i,j)} \ket{\Psi_0} = \tilde\beta_{ij},
		\]
		and such that $|\tilde\alpha_{ij}-\alpha_{ij}/2\Delta_2| < \epsilon / 8\Delta_2 n^2$, and analogously for $\tilde\beta_{ij}$.
	\end{enumerate}
\end{lemma}

We will prove \Cref{lem:ham-sim-props} constructively in \Cref{sec:-tiling,sec:-history};
assuming such a Hamiltonian exists, it is straightforward to see that it also satisfies all the conditions in \Cref{thm:HAexistence}.
\begin{proof}[\Cref{thm:HAexistence}]
	Take $H_A$ from \Cref{lem:ham-sim-props}, and set $P^{(1)}=\Pi\Sq$, $P^{(2)}=\Pi\Flag$, $P^{(3)}=\Pi\Glag$.
	Then $H_A/\Delta(H_A)$ has a spectral gap $\lambda_1-\lambda_0=1$; and by construction the families of real numbers $\{\alpha_{ij}\},\{\beta_{ij}\}$ are approximated to precision $\delta\le \epsilon/4\Delta_2 n^2$, by first rescaling each number by a factor of $2$ (which is where the discrepancy in the bound stems from in \Cref{lem:ham-sim-props}).
\end{proof}

\subsection{Tiling Layers}\label{sec:-tiling}
For a set of tiles $\mathcal T$, we commonly associate a Hilbert space $\hs=\field C^{|\mathcal T|}$ with every site; if the original tiles have a pair of functions
\[
f\Hor:\mathcal T\times \mathcal T \longrightarrow \{0, 1\}
\quad\quad
f\Ver:\mathcal T\times \mathcal T \longrightarrow \{0, 1\}
\]
which determine whether two tiles $t_1,t_2\in \mathcal T$ can lie next to each other on the lattice horizontally or vertically.
Then the corresponding Hamiltonian interaction reads
\[
h\Hor := \1-\sum_{t_1,t_2\in \mathcal T} f\Hor(t_1,t_2)\ketbra{t_1,t_2}
\quad\quad
h\Ver := \1-\sum_{t_1,t_2\in \mathcal T} f\Ver(t_1,t_2)\ketbra{t_1,t_2}.
\]
The terms are simply projectors onto the orthogonal complement of the allowed pairs of tiles, horizontally or vertically; as such, they are two-local by design and act on pairs of qudits of local dimension $|\mathcal T|$.

In addition to constraints on matching tiles, we can introduce one-local on-site projectors with a negative (bonus) or a positive (penalty) coefficient; we assume those are represented by a function $\bar f^{(1)} : \mathcal T \longrightarrow \field R$
which analogously to above translate to the one-local interaction term
\[
h^{(1)} := \sum_{t\in \mathcal T} f^{(1)}(t) \ketbra t.
\]

\begin{lemma}[{\cite[Lem.~1\&Cor.~2]{Bausch2015}}]\label{lem:tiling}
	For a square lattice $\Lambda$ and tileset $\mathcal T$ with tiling rules $f\Hor,f\Ver$ and bonus- and penalty terms $f^{(1)}$, define the tiling Hamiltonian
	\[
	H(\mathcal T) := \sum_{i\in\Lambda} h^{(1)} + \sum_{i\sim j} h\Hor_{i,j} + \sum_{k\sim l} h\Ver_{k,l},
	\]
	where $i\sim j$ sums over horizontally adjacent sites and $k\sim l$ over vertical ones.
	Then $H(\mathcal T)$ is diagonal, where each eigenvector represents a unique tiling configuration from $\mathcal T$ (including mismatching tiles);
	the ground state of $H(\mathcal T)$ is the highest net bonus tiling possible with the weighted tiling rules.
\end{lemma}

We use a straightforward notation for the tiles in the following, where similar to Wang tiles the condition for a compatible placement for adjacent tiles is that their edge colors match.

\subsubsection{Binary Counter}
We define a tiling $\mathcal T_1$ that represents a binary counter, in the sense that the tiling pattern for the $0$\textsuperscript{th} row represents the increasing binary sequence $0_2=\cdots 000, 1_2=\cdots 001, 2_2=\cdots010, \ldots$.
To this end, we follow the construction of \cite{Patitz2014,Bausch2017}, which utilizes seven unique tiles:
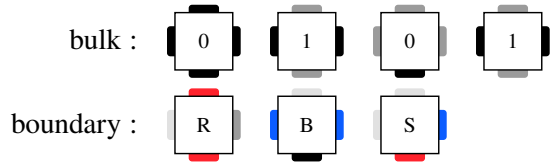
\begin{align*}
\mathrm{bulk:}&\quad
\TTs(0,0,0,0,0)
\quad
\TTs(1,0,1,0,1)
\quad
\TTs(1,1,0,1,0)
\quad
\TTs(1,0,1,0,1)
\\
\mathrm{boundary:}&\quad
\TTs(R,1,R,-,R)
\quad
\TTs(-,B,0,B,B)
\quad
\TTs(-,B,R,-,S)
\end{align*}
In addition to the implicit tiling conditions $f\Hor$ and $f\Ver$ defined by these tiles, we give a bonus of $1/2$ for the S-tile via $\bar f^{(1)}$.
By \Cref{lem:tiling}, this means that the ground state of $H(\mathcal T_1)$ is unique, with an $S$ tile in the upper left corner of the lattice as shown in \Cref{fig:binary-tiling}; this ground state has energy $-1/2$, and all other eigenvalues are $\ge 0$.

\begin{figure}
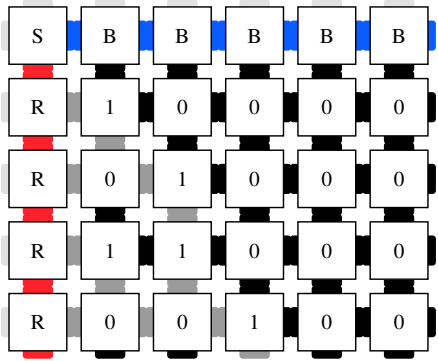

	\centering
	\begin{tiles}*
		\TT(-,B,R,-,S)
		\TT(-,B,0,B,B)
		\TT(-,B,0,B,B)
		\TT(-,B,0,B,B)
		\TT(-,B,0,B,B)
		\TT(-,B,0,B,B)
		
		\$
		
		\TT(R,1,R,-,R)
		\TT(0,0,1,1,1)
		\TT(0,0,0,0,0)
		\TT(0,0,0,0,0)
		\TT(0,0,0,0,0)
		\TT(0,0,0,0,0)
		
		\$
		
		\TT(R,1,R,-,R)
		\TT(1,1,0,1,0)
		\TT(0,0,1,1,1)
		\TT(0,0,0,0,0)
		\TT(0,0,0,0,0)
		\TT(0,0,0,0,0)
		
		\$
		
		\TT(R,1,R,-,R)
		\TT(0,0,1,1,1)
		\TT(1,0,1,0,1)
		\TT(0,0,0,0,0)
		\TT(0,0,0,0,0)
		\TT(0,0,0,0,0)
		
		\$
		
		\TT(R,1,R,-,R)
		\TT(1,1,0,1,0)
		\TT(1,1,0,1,0)
		\TT(0,0,1,1,1)
		\TT(0,0,0,0,0)
		\TT(0,0,0,0,0)
	\end{tiles}
	\caption{Binary counter tiling starting in the upper left corner; the height of the lattice will be written out in binary at the bottom of the lattice, in big Endian bit order.}
	\label{fig:binary-tiling}
\end{figure}

We now add a copy of this tileset $\mathcal T_1'$, mirroring each tile on the diagonal reaching from the lower left to the upper right; the local Hilbert space will then simply be a tensor product $(\field C^{|\mathcal T_1|})^{\otimes 2}$ of the original tiling Hilbert space of dimension $|\mathcal T_1|$.

The result of the two stacked tiling lattices is a binary description of the lattice's height $H$ written out at the bottom left of the lattice, in big Endian order; similarly, the second layer produces a binary description of the lattice width $W$ in binary, on the left edge and starting at the bottom in big Endian order.
Because there are two corner $S$ tiles in the ground state, each with a bonus of $1/2$, the net bonus is $1$.

\subsubsection{Coloring the Triangle and Square}\label{sec:triangle}

\begin{figure}[t]
	\centering
	\begin{tikzpicture}[x=2mm,y=2mm]
	\foreach \x in {0,...,34}{
		\pgfmathsetmacro\yb{34-\x}
		\foreach \y in {0,...,\yb}{
			\fill[yellow!65] (\x,\y) rectangle (\x+1,\y+1);
		}
	};
	\foreach \xy in {0,...,35}{
		\pgfmathsetmacro\yy{35-\xy}
		\fill[yellow] (\xy,\yy) rectangle (\xy+1,\yy+1);
	}
	\fill[green!50] (0,0) rectangle (12,12);
	\foreach \xy in {0,...,11}{
		\pgfmathsetmacro\yy{11-\xy}
		\fill[green!70] (\xy,\yy) rectangle (\xy+1,\yy+1);
	}
	
	\fill[black] (0,35) rectangle (1,36) (0,11) rectangle (1,12);
	\foreach \x in {0,1,3,4,5,8,10,12,13,15,19,20,22,23,24,30}{
		\fill[black] (\x,0) rectangle (\x+1,1);
	};
	
	\draw[step=1,black!40,line width=.5pt] (0,0) grid (40,40);
	
	\draw (-.5,11.5) -- (-2,11.5) node[left] {$2^n$};
	\draw (-.5,35.5) -- (-2,35.5) node[left] {$2^b$};
	\draw [decoration={brace,mirror,amplitude=5}, decorate] (0,-.5) -- (31,-.5) node[midway,below,yshift=-5] {$H$ in binary};
	\draw [decoration={brace,amplitude=5}, decorate] (-6,0) -- (-6,36) node[midway,left,xshift=-5,align=right] {$W=2^n+2^b$\\in binary};
	\end{tikzpicture}
	\caption{Triangle and square tiling pattern.
		The yellow triangle is created by tiles following a diagonal, starting at the location of the binary $1$ on the left lattice edge, offset at position $b$ from the lower left corner.
		Atop the yellow triangle, a similar construction is used to mark a square of size $n \times n$ (shown in light green).
		On the bottom edge, the height of the lattice is written out in binary; it will serve as input to the simulating History State Hamiltonian in \Cref{sec:-history}, which itself runs on the yellow triangle, and simulates a target lattice Hamiltonian within the green square.}\label{fig:triangle}
\end{figure}
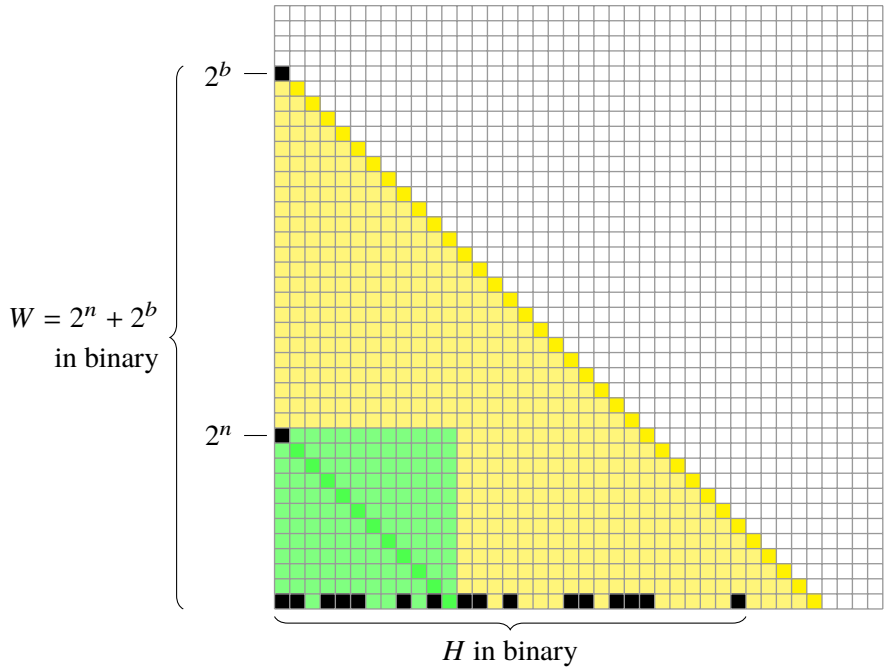

In order to create the triangular and square shapes outlined in \Cref{sec:construction-overview}, we define two tilesets that are defined on the background laid out by $\mathcal T_1$ and $\mathcal T_1'$ in the last section.
We leave the concrete set of tiles as an exercise to the reader, but emphasize that they are straightforward to write down.
The first new tileset, $\mathcal T_2$, will be used to draw a diagonal yellow line starting from the $1$ on the left lattice edge, offset at $b$ from the lower left edge, and continuing all the way down to the lower lattice edge.
The tiles underneath the yellow line are constrained to be marked in a yellow shade as well.

Similarly, $\mathcal T_2'$ is used to first draw a diagonal green line starting on the left edge at offset $n$ from the lower left lattice corner; the area below is shaded in green.
To achieve that the entire square of side length $n \times n$ is marked in green, we can for instance single out the first $n$ rows and $n$ columns and constrain them each to be in a marked configuration; the square is than simply the intersection of these two rectangles.

\subsection{History State Layers}\label{sec:-history}
We assume we are given a universal classical Turing machine $M$ over a set of internal states $Q$ and tape alphabet $A$.
This classical TM will serve as a ``clock'', i.e.\ as the automaton driving a quantum Thue system (QTS, see \cite[Def.~52]{Bausch2016}), which in turn will perform the Hamiltonian simulation on $\Lambda\Sim$.

To this end, we add a Hilbert space $\field C^d$ to each lattice vertex, large enough to contain all transition rules for the Turing machine and quantum Thue system.
This model of computation is known to be universal for the uniform family of quantum circuits constrained only by the available space and time; i.e.\ in our case, where the tape space is dictated by the size of the available lattice this computation will be space-bounded, with an up to exponentially-long runtime in the available space (cf.~\cite[Sec.~6.3\&Lem.~59]{Bausch2016}).

We also know that once translated to a history state Hamiltonian, we can analyse and treat the Hamiltonian as a unitary labelled graph (ULG, see \cite[Def.~35]{Bausch2016}), which by \cite[Th.~10]{Bausch2016a} exhibit a spectral gap shrinking inverse quadratically in the graph's diameter---under the condition that the ULG representing the ground state of the associated Hamiltonian is \emph{simple}, which on the side of the encoded computation simply means that there are no loops in the clock driving the QTS which lead to incompatible outcomes; as we can choose the classical Turing machine that serves as clock this is guaranteed in our case.

The Hamiltonian associated to the QTS---which in turn represents a quantum computation $U = U_TU_{T-1}\cdots U_2U_1$---is a so-called \emph{history state Hamiltonian} $H\Hist$, with ground state spanned by states of the form
\begin{equation}\label{eq:history-state}
\ket{\Psi_0} = \frac{1}{\sqrt{T-1}} \sum_{t=0}^T \ket t\ket{\psi_t} \in \hs_c \otimes \hs_q,
\quad\text{where}\quad
\ket{\psi_t} := U_t U_{t-1}\cdots U_2U_1 \ket{\psi_0},
\end{equation}
\newcommand\In{^{\mathrm{in}}}
for some initial state $\ket{\psi_0} \in \hs_q$.
By adding a penalty term of the form $\ketbra 0 \otimes \Pi\In$, where $\Pi\In = \1-\ketbra{\psi\In}$ 
for some initial configuration $\ket{\psi\In} \in \hs_q$, we can ensure that the groud state of $H\Hist$ is unique and such that $\ket{\psi_0} = \ket{\psi\In}$---i.e.\ the computation encoded in the ground state is correctly initialized.

In order to define the QTS transition rules and for consistency of notation, we will define the following set of tiles, representing two copies of the Turing machine symbols $a_i \in A$:
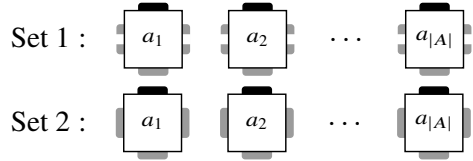
\begin{align*}
\mathrm{Set\ 1:}&\quad
\TTs(0,*1,1,*1,$a_1$)
\quad
\TTs(0,*1,1,*1,$a_2$)
\quad
\cdots
\quad
\TTs(0,*1,1,*1,$a_{|A|}$)
\\
\mathrm{Set\ 2:}&\quad
\TTs(0,1,1,1,$a_1$)
\hspace{.2mm}
\quad
\TTs(0,1,1,1,$a_2$)
\quad
\cdots
\quad
\hspace{.2mm}
\TTs(0,1,1,1,$a_{|A|}$)
\end{align*}
The head of the Turing machine will sit between tiles from Set 1 and 2, by defining, for each internal state $q_i\in Q$, the head tiles
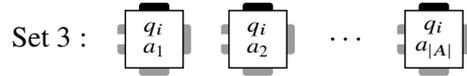
\[
\mathrm{Set\ 3:}\quad
\TTs(0,1,1,*1,\TMstate{q_i}{a_1})
\quad
\TTs(0,1,1,*1,\TMstate{q_i}{a_2})
\quad
\cdots
\quad	
\TTs(0,1,1,*1,\TMstate{q_i}{a_{|A|}})
\]
such that the head state lies above the tape.
For instance, for a binary alphabet, one particular time configuration of said Turing machine could be
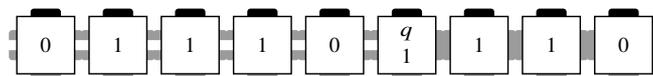
\[
\begin{tiles}*
\TT(0,*1,1,*1,0)
\TT(0,*1,1,*1,1)
\TT(0,*1,1,*1,1)
\TT(0,*1,1,*1,1)
\TT(0,*1,1,*1,0)
\TT(0,1,1,*1,\TMstate{q}{1})
\TT(0,1,1,1,1)
\TT(0,1,1,1,1)
\TT(0,1,1,1,0)
\end{tiles}
\]
In this way, horizontally, and by construction, there can only ever be a single TM head; this can be enforced with a regular expression as in \cite{Gottesman2009}.

If the local tile space is again defined as $\hs_c\otimes \hs_q$ for the overall tileset $T$ and quantum space $\hs_q$, a transition for this Turing machine from two neighbouring tiles $t_1t_2$ to tiles $s_1s_2$ is then simply defined by a local Hamiltonian term
\[
h := (\ket{t_1t_2}-\ket{s_1s_2})(\bra{t_1t_2}-\bra{s_1s_2}) \otimes \1
\]
Were we to allow the QTS to perform a quantum operation on $\hs_q$, e.g.\ a unitary $U\in SU(\hs_q)$, then we write
\[
h_{U} := \left(\ketbra{t_1t_2} + \ketbra{s_1s_2}\right)\otimes \1-\ketbra{t_1t_2}{s_1s_2}\otimes U^\dagger-\ketbra{s_1s_2}{t_1t_2}\otimes U.
\]
All this is standard and well-studied; it is obvious how to translate the TM's evolution within a horizontal section of the lattice into such transition terms.

What we need to address explicitly is how to map a one-dimensional TM to two dimensions, i.e.\ onto the yellow triangular subsection of $\Lambda\Sim$ as shown in \Cref{fig:triangle}, as we need to write out a specific two-dimensional pattern of phases as required by \Cref{lem:ham-sim-props}.
In order to achieve this, note how we specified the TMs tiles such that the top color was always black and the bottom one grey.
We duplicate the tiles with those colors switched, e.g.\
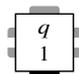
\[
\TTs(1,*1,0,*1,\TMstate{q}{1})
\]
We write the transition terms such that the TMs direction of ``left'' and ``right'' are interchanged, depending on whether the head is currently on an even or odd layer.
What is left is to specify how the TM head can be turned around, should it bounce into a specific boundary tile.

This is straightforward; if the red vertical edge delineates the left boundary of the triangular section in \Cref{fig:triangle}, we add the two-local transition terms
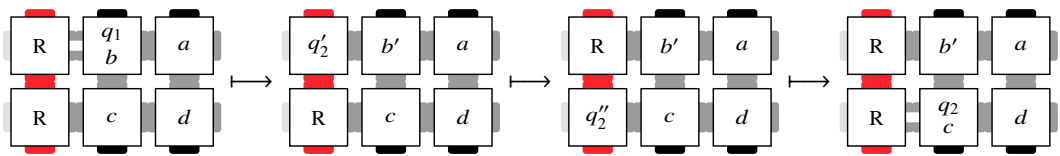
\[
\begin{tiles}*
\TT(R,*1,R,-,R)
\TT(0,1,1,*1,\TMstate{q_1}{b})
\TT(0,1,1,1,$a$)
\$
\TT(R,1,R,-,R)
\TT(1,1,0,1,$c$)
\TT(1,1,0,1,$d$)
\end{tiles}
\longmapsto
\begin{tiles}*
\TT(R,1,R,-,$q_2'$)
\TT(0,1,1,1,$b'$)
\TT(0,1,1,1,$a$)
\$
\TT(R,1,R,-,R)
\TT(1,1,0,1,$c$)
\TT(1,1,0,1,$d$)
\end{tiles}
\longmapsto
\begin{tiles}*
\TT(R,1,R,-,R)
\TT(0,1,1,1,$b'$)
\TT(0,1,1,1,$a$)
\$
\TT(R,1,R,-,$q_2''$)
\TT(1,1,0,1,$c$)
\TT(1,1,0,1,$d$)
\end{tiles}
\longmapsto
\begin{tiles}*
\TT(R,1,R,-,R)
\TT(0,1,1,1,$b'$)
\TT(0,1,1,1,$a$)
\$
\TT(R,*1,R,-,R)
\TT(1,1,0,*1,\TMstate{q_2}{c})
\TT(1,1,0,1,$d$)
\end{tiles}
\]
that correspond to a ``left'' move of the TM when it is in state $q_1$, reads $b$, and writes $b'$.
Similar terms can easily be added for switching layers on the right side of the tape, and it is obvious how this can be extended to non-straight boundaries (e.g.\ for the triangular patch constructed in \Cref{sec:triangle}).

The history state Hamiltonian encoding a given QTS $Q$ we denote with $H\Hist(Q)$.

\subsection{Grid Simulator QTS}
\begin{lemma}\label{lem:grid-qts}
	Let $n\in\field N$, $\epsilon>0$, and take two families of real numbers $A=\{ \alpha_{ij} \}_{i,j\in\{1,\ldots,n\}}$, $B=\{ \beta_{ij} \}_{i,j\in\{1,\ldots,n\}}$.
	Each $\alpha_{ij},\beta_{ij}\in [0,\Delta_2]$, as well as $\Delta_2$ and $\epsilon$, have bit precision $\Xi\in\field N$.
	Then there exists $b'\in\field N$, $b'=\poly(n,\Xi)$, and a $2$-local poly-time terminating QTS with space bound $b'$, which on input $(A,B,\Delta_2,\epsilon)$ written out in binary on the initial string, writes out a quantum state
	\[
	\ket{\Theta} := \bigotimes_{i=1}^b \bigotimes_{j=1}^i \begin{cases}
	\ket{\tilde\alpha_{ij}} \ket{\tilde\beta_{ij}} \ket{s_{i,j}} & i,j\le n \\
	\ket{r_{i,j}} & \text{otherwise},
	\end{cases}
	\]
	where $\| \ket*{\tilde\alpha_{ij}}-\ket*{\alpha_{ij}}/\Delta_2 \| \le \epsilon / 4\Delta_2n^2$, and where $\ket*{\alpha_{ij}}$ is defined in \Cref{eq:theta-state}; analogously for $\ket*{\tilde\beta_{ij}}$; and the $\ket*{r_{i,j}},\ket*{s_{i,j}}$ remain unspecified.
\end{lemma}
\begin{proof}
	We first note that the input size $|(A,B,\Delta_2,\epsilon)| \le (2n^2 + 2)\Xi = \poly(n,\Xi)$.
	By the Solovay-Kitaev theorem, any rotation gate $R_{i,j}\ket 0 = \ket{\alpha_{ij}}$ can be approximated to precision $\delta$ in $T = \BigO(\log^4 \delta^{-1})$ many steps from a finite universal gateset \cite{Dawson2005}; where we require
	\[
	\delta = \frac{\epsilon}{4\Delta_2 n^2}
	\quad\text{and thus}\quad
	T = \BigO\left( \log^4\frac{\Delta_2 n^2}{\epsilon} \right) = \BigO\left( \log^4 n, \log^4 \frac1\epsilon, \Xi^4 \right) = \poly(\log^4 n, \Xi)
	\]
	as necessarily $\Delta_2,\epsilon^{-1} \le 2^\Xi$.
	Since we need to perform $2n^2$ of these rotations, overall we require $T' = \poly(n,\Xi)$ many steps.
	The rest of the claim follows from the fact that quantum Thue systems are a universal computational model, even when restricted to two-local rules, and by choosing $b=\poly n$ large enough for the at most polynomial runtime and space overhead from implementing the computation as a quantum Thue system.
\end{proof}

\begin{corollary}\label{cor:H_A-existence}
	Define
	\[
	H_A = H(\mathcal T_1) + H(\mathcal T_1') + H(\mathcal T_2) + H(\mathcal T_2') + H\Hist(Q),
	\]
	where $\mathcal T_1,\mathcal T_1',\mathcal T_2,\mathcal T_2'$ are the tilesets defined in \Cref{sec:-tiling}, and $H(\mathcal T)$ for a tileset is given in \Cref{lem:tiling}; $H\Hist(Q)$ is the history state Hamiltonian implementing the QTS from \Cref{lem:grid-qts}, as defined in \Cref{sec:-history}.
	Then $H_A$ satisfies \Cref{lem:ham-sim-props}.
\end{corollary}
\begin{proof}
	Both the triangular and square colored region in \Cref{fig:triangle} can be associated with local projectors $\Pi\Tri,\Pi\Sq$.
	In order to obtain the required local projectors $\Pi\Flag$ and $\Pi\Glag$, we define them for each lattice site $(i,j)\in\Lambda$ via
	\[
	\Pi\Flag^{(i,j)} := \ketbra{T} \otimes \ketbra{i,j}\Flag
	\quad\text{and}\quad
	\Pi\Glag^{(i,j)} := \ketbra{T} \otimes \ketbra{i,j}\Glag,
	\]
	where $\ketbra{T}$ signals that the QTS computation from \Cref{lem:grid-qts} is done (which is locally checkable), and $\ket{i,j}\Flag$ or $\ket{i,j}\Glag$ are a specific local state we ``blink'' on and off for precisely \emph{half} of the time steps in the computation (for an explicit construction see \cite[Lem.~16]{Bausch2018b}).
	
	The Hamiltonian associated to the QTS is a history state Hamiltonian, as described at the start of \Cref{sec:-history}.
	The overlap of $\Pi\Flag$ and $\Pi\Glag$ thus equals $(T/2)/T=1/2$, and for the history state $\ket{\Psi_0}$ we get
	\[
	\bra{\Psi_0} \Pi\Flag^{(i,j)} \ket{\Psi_0} = \frac{\tilde\alpha_{ij}}{2}
	\quad\text{and}\quad
	\bra{\Psi_0} \Pi\Glag^{(i,j)} \ket{\Psi_0} = \frac{\tilde\beta_{ij}}{2}.
	\]
	By \cite[Th.~10]{Bausch2016a} the spectral gap $\Delta(H_A) \sim 1/T'^2 = 1/\poly b$.
	The rest of the claim follows by construction.
\end{proof}

\section{Simulating Translationally-Invariant Hamiltonians}\label{sec:TItarget}
In order to simulate a translationally-invariant Hamiltonian, it suffices to approximate the \emph{local} terms in a periodic fashion over the underlying square lattice $\Lambda\Sim$.
This is straightforward in our current construction, and will mostly hinge on modifying the QTS in the history state to be dovetailed by a computation that simply \emph{repeats} the creation of the family of $k\times k$ numbers $\{\alpha_{ij}\},\{\beta_{ij}\}$ in a periodic fashion---i.e.\ for $n\times n$ blocks of size $k\times k$---across the entire lattice (instead of being restricted to the lower left corner of $\Lambda\Sim$).

Without going through the entire construction in rigorous detail as before, we give a proof outline in the following, and leave the gaps to be filled by the reader.
We start with a modification of \Cref{lem:grid-qts} which adds a dovetailed procedure that loops the creation of the field of real numbers across the entire lattice grid $\Lambda\Sim$.

\begin{lemma}\label{lem:grid-qts-TI}
	Let $n\in\field N$, $\epsilon>0$, and take two families of real numbers $A=\{ \alpha_{ij} \}_{i,j\in\{1,\ldots,k\}}$, $B=\{ \beta_{ij} \}_{i,j\in\{1,\ldots,k\}}$ for some constant $k\in\field N$.
	Each $\alpha_{ij},\beta_{ij}\in [0,\Delta_2]$, as well as $n$, $\Delta_2$ and $\epsilon$, have bit precision $\Xi\in\field N$.
	Then there exists a $2$-local poly-time terminating QTS, which on input $(A,B,\Delta_2,\epsilon,n)$ written out in binary on the initial string, writes out a quantum state
	\[
	\ket{\Theta} := \bigotimes_{(i,j)\in\Lambda\Sim} \begin{cases}
	\ket{\tilde\alpha_{i\bmod n,j\bmod n}} \ket{\tilde\beta_{i\bmod n,j\bmod n}} \ket{s_{i,j}} & \text{for\ } i,j\le nk \\
	\ket{r_{i,j}} & \text{otherwise},
	\end{cases}
	\]
	where $\| \ket*{\tilde\alpha_{ij}}-\ket*{\alpha_{ij}}/\Delta_2 \| \le \epsilon / 4\Delta_2n^2$, and where $\ket*{\alpha_{ij}}$ is defined in \Cref{eq:theta-state}; analogously for $\ket*{\tilde\beta_{ij}}$; and the $\ket*{r_{i,j}}$ remain unspecified.
\end{lemma}
\begin{proof}
	The QTS performs the same operations as in \Cref{lem:grid-qts} to recreate a $k\times k$ grid of real numbers $\{\tilde\alpha_{ij}\},\{\tilde\beta_{ij}\}$; yet once it is done, it copies the (classical) input $(A,B,\Delta_2,\epsilon)$ to a new $k\times k$ grid square and recreates the same families of angles there.
	This is repeated until $n\times n$ patches of $k\times k$ numbers are filled.
	
	While we will require an overhead to create a single $k\times k$ patch that will protrude into \emph{other} patches, we can always demand that the QTS uncomputes the slack space and restores it to a clean ancilla state for the computation of the next patch.
	
	The runtime of the entire procedure is upper-bounded by $\poly(T'n^2)$, where $T'=\poly(k,\Xi)$ is the runtime from \Cref{lem:grid-qts} to create a single $k\times k$ patch of numbers.
\end{proof}

This modified computation allows us to formulate a variant of \Cref{lem:ham-sim-props} for simulating a translationally-invariant target system.
\begin{lemma}\label{lem:ham-sim-props-ti}
	Let $n\in \field N$, $\epsilon>0$, and take two families of real numbers $\{\alpha_{ij}\}_{i,j\in \{1,\ldots, k\}}$ and $\{\beta_{ij}\}_{i,j\in\{1,\ldots,k\}}$, for $k\in\field N$ constant.
	Each $\alpha_{ij},\beta_{ij}\in[0,\Delta_2]$, as well as $\Delta_2$ and $\epsilon$, have bit complexity $\Xi\in\field N$.
	Then the following exists.
	A spin lattice $\Lambda\Sim$ of side length $W\times H$, and interactions $h^{(1)},h\Hor,h\Ver$, 
	\begin{enumerate}
		\item $H=\exp \poly(k,\Xi)$
		\item $W=2^k + 2^b$ for some $a,b\in\field N$, $k<b/2$,
		\item $h^{(1)},h\Hor,h\Ver$ are independent of $W$ and $H$,
	\end{enumerate}
	and with matrix entries and $H_A$ as defined in \Cref{lem:ham-sim-props}.
	Denote with $\Pi\Flag,\Pi\Glag$ two projectors onto two distinct local basis states.
	Define $H_A$ as in \Cref{lem:ham-sim-props}.
	Then
	\begin{enumerate}
		\item[1.] $\Delta(H_A) \ge 1/\poly(n,k,\Xi)$.
		\item[6.] for all $(i,j)\in\Lambda\Sim$, we have for all $i,j\le nk$,
		\[
		\bra{\Psi_0} \Pi\Flag^{(i,j)} \ket{\Psi_0} = \tilde\alpha_{i\bmod k,j\bmod k}
		\quad\text{and}\quad
		\bra{\Psi_0} \Pi\Glag^{(i,j)} \ket{\Psi_0} = \tilde\beta_{i\bmod k,j\bmod k},
		\]
	\end{enumerate}
	and such that $|\tilde\alpha_{ij} - \alpha_{ij}/2\Delta_2| < \epsilon/8\Delta_2 n^2$, and analogously for $\tilde\beta_{ij}$.
\end{lemma}
\begin{proof}
	We take the TM constructed in \Cref{sec:-tiling,sec:-history} to calculate a $k\times k$-sized patch of coefficients given by the two families of real numbers, and encoded by the side length of the lattice.
	The difference is that the QTS constructed in \Cref{lem:grid-qts-TI} periodically creates this $k\times k$ patch of numbers across the entire lattice $\Lambda\Sim$ by \emph{ignoring} the underlying tiling square and triangle, and simply recreates all the encoded numbers $(\alpha_{i',j'},\beta_{i',j'})$ on sites $i'=i\bmod k$, $j'=j\bmod k$.
	The rest of the claim follows as in \Cref{lem:ham-sim-props}.
\end{proof}

We emphasize that the local dimension of the lattice model is---while finite---not restricted; as such, any overhead due to the necessarily more-complicated encoded QTS can be captured by simply increasing the local dimension of the simulating spin lattice $\Lambda\Sim$.

As a consequence, we get the following corollary, vaguely phrased on purpose.
\begin{corollary}
	There exist universal translationally-invariant Hamiltonians which can simulate a translationally-invariant $H\Tar$ with an at most polynomial spatial overhead.
\end{corollary}

\section{Qubit Interactions}
\label{sec:qubits}
In this section we consider qubit interactions restricted to the edges of a translationally-invariant interaction graph. First we extend the work of \cite{Oliveira-Terhal} to higher dimensions, introducing the notion of geometric locality, and showing that interactions on a (hyper-)cubic lattice can efficiently simulate all geometrially local Hamiltonians in the same number of spatial dimensions.

Then we go further to show that the same result holds for simulator Hamiltonians restricted to any translationally-invariant graph.

\subsection{Spatial Sparsity and Geometric Locality}

\label{sec:geometriclocality}
In \cite{Oliveira-Terhal}, Oliveira and Terhal proved that the local Hamiltonian problem is QMA-complete, even for 2-local Hamiltonians with an interaction graph 
to the edges of a square lattice.
They did this by first showing directly that the local Hamiltonian problem for Hamiltonians with a spatially sparse interaction graph is QMA-complete, where they defined a spatially sparse interaction graph as in Definition~\ref{def:spatialsparsity}.
They then showed that all spatially sparse local Hamiltonians can be simulated by Hamiltonians on a square lattice using a constant number of rounds of perturbation theory -- thereby keeping the strength of the interactions $\BigO(\poly(n))$.

We will identify some implicit assumptions in the definition of spatial sparsity, which are required in order for the proof that they can be simulated by 2D square lattice Hamiltonians to go through.

\begin{definition}[Spatial sparsity~\cite{Oliveira-Terhal}]
	\label{def:spatialsparsity}
	A spatially sparse interaction (hyper-)graph $G$ is defined as a (hyper-)graph in
	which 
	
	\begin{enumerate}[label=\roman*)]
		\item every vertex participates in $\BigO(1)$ hyperedges, 
		\item there is a straight-line drawing in the
		plane such that every hyperedge overlaps with $\BigO(1)$ other hyperedges \item the surface covered by
		every hyperedge is $\BigO(1)$.
		\label{cond:surface}
	\end{enumerate}
\end{definition}

We observe that this definition seems a bit problematic, in particular condition \ref{cond:surface} ``the surface covered by
every hyperedge is $\BigO(1)$'' (for hyperedges involving only two vertices (i.e. a normal edge), we interpret condition \ref{cond:surface} to mean ``the \emph{length} of every edge is $\BigO(1)$''). 
There are two problems here: first, any graph can trivially satisfy this condition simply by rescaling the entire graph and pushing all the vertices very close together (and this will not affect the other conditions i) and ii)). 
This clearly seems like ``cheating'', and there must be an implicit assumption that this is not allowed. 
Indeed this assumption is crucial when simulating a spatially sparse Hamiltonian on a square lattice, because the first step is to place a square grid over the graph, and then snap vertices to the nearest grid point, choosing a small enough grid spacing such that each vertex moves to a unique grid point.
This is not possible with $\Omega(1)$ grid spacing if there are many vertices very close together in an area of $\mathrm o(1)$ size.

Second, even with this implicit assumption, this condition does not seem to match what can be simulated with a constant number of rounds of perturbation theory.
Consider an interaction between three vertices located in the plane at coordinates $(-n,0)$, $(0,1/n)$, and $(n,0)$. The surface covered by this hyperedge is the area of this triangle $\frac{1}{2}2n \times \frac{1}{n}=1$, but the distance between each pair of vertices is $\ge n$. Therefore simulating this interaction with the subdivision gadgets we have will take more than a constant number of rounds of perturbation theory. 
To avoid this problem, the condition should be rephrased as ``each hyperedge is contained in a ball of $\BigO(1)$ radius'' or equivalently ``each qudit only interacts with other qudits at most $\BigO(1)$ distance away''.

We therefore propose the following alternative property, which we call \emph{geometric locality}, in order to avoid these problems as well as generalising the property to dimensions of Euclidean space higher than 2.

\begin{definition}[Geometric locality]
	\label{def:geometriclocality}
	A collection of (hyper-)graphs embedded in $\R^D$ are \emph{geometrically local} if there exist constants $c,C \in \R$ such that for each (hyper-)graph,
	\begin{enumerate}[label=\roman*)]
		\item there are no more than $c$ vertices in any ball of radius 1, and
		\item there are no (hyper-)edges between vertices that are more than distance $C$ apart.
		\end{enumerate}

\end{definition}

We say that a family of Hamiltonians is geometrically local in $\R^D$, if their interaction (hyper-)graphs are geometrically local in $\R^D$.
The definition of geometric locality therefore captures the idea of a Hamiltonian without long-range interactions by condition ii), while ensuring that one cannot ``cheat'' by putting a large number of particles into the same space by condition i).

Notice that the first two conditions of the spatial sparsity definition now follow from the definition of geometric locality! Since each vertex in a geometrically local hypergraph only participates in hyperedges with vertices in an $\BigO(1)$ sized ball and there are at most $\BigO(1)$ vertices in any $\BigO(1)$ sized ball, each vertex can participate in at most $\BigO(1)$ hyperedges. 
And each hyperedge can intersect at most $\BigO(1)$ other hyperedges.
Furthermore each hyperedge can involve at most $\BigO(1)$ vertices (another important condition, assumed implicit when discussing local Hamiltonians in the definition of spatial sparsity).

\subsection{Hypercubic Lattice}
\begin{lemma}
	\label{lem:GL=cubic}
	Let $\cS$ be either $\{XX+YY\}$ or $\{XX+YY+ZZ\}$. Then the family of Hamiltonians with geometrically local interaction hypergraphs embedded in $\R^D$ (for $D\ge2$) can be efficiently simulated by $\cS$-Hamiltonians, even when the interaction graph of the simulator Hamiltonian is a degree $3$ subgraph of the hypercubic lattice in $\R^D$.
\end{lemma}

\begin{proof}
	We first note that any geometrically local family of Hamiltonians must be $k$-local for $k=O(1)$, since each qudit can only interact with those in an $\BigO(1)$ ball around it, and there can only be $\BigO(1)$ other qudits in such a ball.
	By Theorem~\ref{thm:classification}, we know that $\cS$-Hamiltonians can efficiently simulate $k$-local Hamiltonians, when $k=O(1)$. Furthermore, since the proof method only uses local gadgets, the interaction graph of the simulator Hamiltonian is also geometrically local in $\R^D$.
	
	Then it will suffice to simulate geometrically local $\cS$-Hamiltonians. The following argument adapts the proof in \cite{Oliveira-Terhal} for the $D=2$ case, and uses the perturbative gadgets from \Cref{lem:PMgadgets}.
	We use the subdivision gadget on each edge to isolate high degree vertices. Then the degree of a vertex of degree $d>3$ can be reduced by using the fork gadget $\BigO(\log d)$ times in series. Doing this for all vertices in parallel results in a simulator Hamiltonian with a degree 3 geometrically local interaction graph in $\R^D$.
	
	Now place a fine (hyper-)cubic lattice over the interaction graph in $\R^D$, and snap vertices to the nearest grid point.
	Since there at most $\BigO(1)$ vertices in a ball of constant volume, at most $\BigO(1)$ vertices will snap to each grid point. 
	Therefore decreasing the grid spacing by a constant multiplicative factor will allow all of the vertices to be locally reassigned to a unique point on the hypercubic lattice. Each vertex has only moved by a constant distance so the interaction graph remains geometrically local.
	
	We now snap each edge in the interaction graph to a path along edges of the hypercubic lattice. We intend to simulate this interaction using multiple applications of the subdivision gadget along this path, but first we must ensure that no two paths overlap.
	Because there are at most $\BigO(1)$ edges in a ball of constant volume, each path can overlap with at most $\BigO(1)$ other paths. 
	By decreasing the grid spacing by a further constant factor if necessary, there will be enough space to locally re-route each path so that no two paths overlap.
	
	We now simulate each interaction along one of these paths using multiple applications of the subdivision gadget.
	Since each path is at most $\BigO(1)$ in length, this takes only $\BigO(1)$ rounds of perturbation theory, and so the weights are at most polynomially large.
	
	Finally we note that when $D=2$, it will not always be possible to reroute paths so that they do not cross. 
	This problem can be solved using the crossing gadget from \Cref{lem:PMgadgets}.
	
\end{proof}
\subsection{Translationally-Invariant Graphs}
\label{sec:TIgraphs}
We have shown that $\cS$-Hamiltonians can efficiently simulate all geometrically local Hamiltonians in $\R^D$, even when restricted to the edges of the hypercubic lattice in $\R^D$.
We will show that there is nothing special about the hypercubic lattice, and that this result holds for any connected, translationally-invariant graph in $\R^D$.

\begin{definition}
	Let $G$ be a graph embedded in $\R^D$. We say that $G$ is \emph{translationally invariant} if there exists a basis $\{\ket{v_i}\}_{i=1}^D$ for $\R^D$ such that $G$ is invariant with respect to a translation by $\ket{v_i}$ for any $i \in \{1,\dots ,D\}$.
\end{definition}

This notion of translationally-invariant graphs is slightly more general than the standard definition of a lattice in which every vertex must be at a location $\sum_{i=1}^D k_i \ket{v_i}$, for some basis $\ket{v_i}$ where all $k_i \in \Z$.

Before we prove the main result of this section about simulating generic Hamiltonians on translationally-invariant graphs, we must first prove the following lemma about translationally-invariant graphs in $\R^D$.
Recall that a graph $H$ is a minor of $G$ if $H$ can be obtained by deleting edges and vertices and contracting edges.
\begin{lemma}
	\label{lem:TIgraph=cubic}
	Let $G$ be a connected translationally-invariant graph embedded in $\R^D$. Then the hypercubic lattice is a minor of $G$. Furthermore, this happens in a regular way: there exists a connected subgraph $T$ and a basis $\{\ket{w_i}\}_{i=1}^D$ such that if each translation of $T$ with respect to $\{\ket{w_i}\}_{i=1}^D$ is contracted to a single vertex, we obtain a graph which has the hypercubic lattice as a subgraph.
	There is a deterministic algorithm to find such a $T$ and the basis $\{\ket{w_i}\}_{i=1}^D$.
	
\end{lemma}

Figure~\ref{fig:hexagonal} shows how Lemma~\ref{lem:TIgraph=cubic} applies to the hexagonal lattice in 2D.
Here $T$ contains just two vertices. Each quadrilateral with dashed edges contains a translation of $T$.
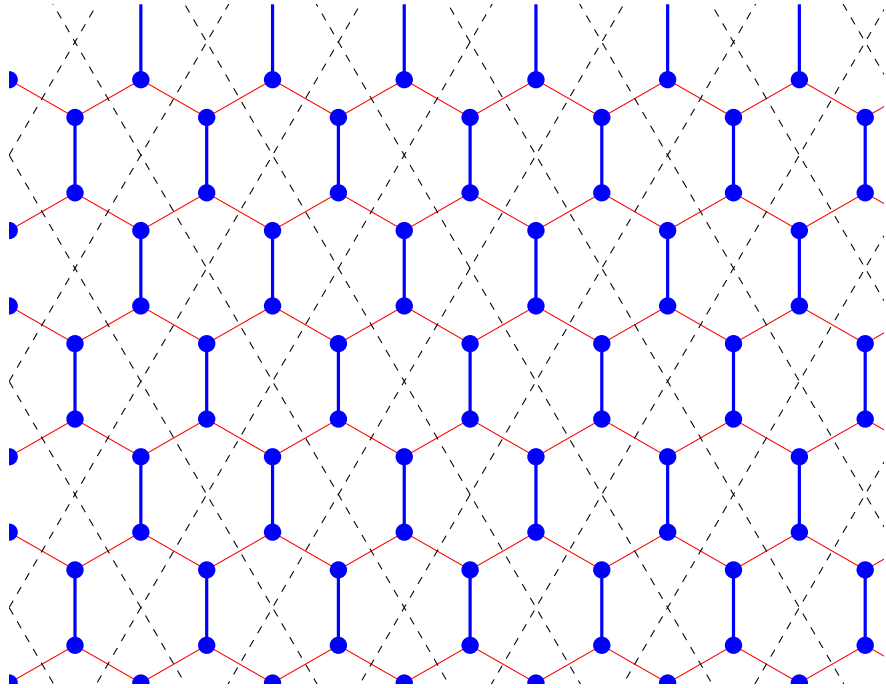
\begin{figure}
	\centering
	\begin{tikzpicture}[scale=0.5]
	
	\pgfmathsetmacro{\x}{6}
	\pgfmathsetmacro{\y}{2}
	\clip (0,0) rectangle (23,18);
	\foreach \j in {0,...,\y}{
		\begin{scope}[shift={(0,6*\j)}]
		
		\foreach \i in {0,...,\x}{
			\begin{scope}[shift={({2*sqrt(3)*\i},0)}]
			\draw[red] (0,0) -- (30:2)-- ++(330:2);
			\draw[red] (0,4) -- ++(330:2)-- ++(+30:2);
			\draw[blue,very thick] ({sqrt(3)},1) -- ++(0,2);
			\node[qubit,blue] at (0,0){};
			\node[qubit,blue] at (0,4){};
			\begin{scope}[shift={(30:2)}]
			\node[qubit,blue] at (0,0){};
			\node[qubit,blue] at (0,2){};
			\end{scope}
			\end{scope}}
		
		\foreach \i in {1,...,\x}{
			\begin{scope}[shift={({2*sqrt(3)*\i},0)}]
			\draw[blue,very thick] (0,4)-- (0,6);
			\end{scope}}

		\end{scope}
	}
	
	\foreach \i in {-10,...,3}{
		\draw[dashed] (0,2+6*\i) -- ++(60:100);
		\draw[dashed] (0,2-6*\i) -- ++(-60:100);}

	\end{tikzpicture}
	\caption[Identification of hexagonal lattice with square lattice]{The hexagonal lattice is isomorphic to a square lattice where each vertex has been replaced by two vertices in a subgraph $T$. Each translation of $T$ is coloured blue, and edges between translations of $T$ are coloured red.}
	\label{fig:hexagonal}
\end{figure}

\begin{proof}
	We identify each vertex $x \in G$ with the vector $\ket{x}\in \R^D$ representing its location in the embedding in $\R^D$.
	
	Note that the first statement, that the hypercubic lattice is a minor of $G$, follows from the rest of the Lemma.
	Let $\{\ket{v_i}\}_{i=1}^{D}$ be the basis of $\R^D$ with respect to which the graph is invariant. We will show how to find an alternative basis $\{\ket{w_i}\}_{i=1}^{D}$ and a finite connected subgraph $T$, such that the graph is also invariant with respect to this basis, and that for all $i \in \{1,\dots , D\}$, $T$ and $T+\ket{w_i}$ have one edge connecting them and that 
	\[T \bigcap \left(T+\sum_{i=1}^D k_i \ket{w_i}\right) = \emptyset \quad \text{ for any } (k_1,\dots,k_D) \in \Z^D/\{0^D\} \]
	
	We prove that this subgraph $T$ exists by induction on $j,$ the number of $i$ for which this condition holds with each step.
	
Pick any vertex $u \in G$ and let $T_0=u$.
 Now let $0\le j\le D$, and assume that we have a linearly independent set of vectors $\{\ket{w_i}\}_{i=1}^{j}$ and a connected graph $T_j$ such that $T_j$ is disjoint with $T_j +\sum_{i=1}^{j} k_i \ket{w_i}$ for all $k \in \Z^{ j} / \{0^{ j}\}$; and that there is an edge between $T_j$ and $T_j + \ket{w_i}$ for all $i \le j$.
		
		Let $S_j$ be the surface effectively spanned by the $T_j$, defined as:
		\[S_j= \bigcup_{k \in \Z^j}\left(T_j + \sum_{i=1}^{j} k_i \ket{w_i}\right)\]
		%$Let $M \subseteq \Z$ be the set of integers $m$ such that $S_j$ does not intersect with $S_j + m \ket{v_{j+1}}$.
and let $s_j\ge 0$ be the largest integer such that $S_j$ intersects $S_j +\ket{v_{j+1}}$.
		
		Find a shortest path $p_j$ between $T_j$ and $\bigcup_{m >s_j}\left(S_j + m\ket{v_{j+1}}\right)$. 
		%Then there exists a path $p_j$ of length $d_j$ between $T_j$ and $S_j +m\ket{v_{j+1}}$ for some $m$. 
		Let $x$ be the first vertex along this path (starting from $T_j$) which intersects with a translation of either $p_j$ or $T_j$ in the positive $\ket{v_{j+1}}$ direction; that is, $x$ is a vertex which is in $p_j$ and also in \[(p_j\cup T_j)+ l_{j+1}\ket{v_{j+1}}+\sum_{i=1}^{j} l_i \ket{w_{i}} \text{ for some } l \in \Z^{j+1} \text{ with } l_{j+1}>s_j.\]
		
		Then let $p'_j$ be the path $p_j$ up to but not including the vertex $x$ (in the direction starting from $T_j$) and  set
		\[T_{j+1}=T_j\cup p'_j \quad \text{ and } \quad \ket{w_{j+1}}= l_{j+1}\ket{v_{j+1}}+\sum_{i=1}^{j} l_i \ket{w_{i}}.\]
		There is an edge between $T_{j+1}$ and $T_{j+1}+\ket{w_i}$ for all $i\le j$, by the induction hypothesis that the same condition holds for $T_j$.
		By definition of $p'_j$, there is also an edge from $T_{j+1}$ to $T_{j+1}+\ket{w_{j+1}}$.
		
We need to check that $T_{j+1}$ does not intersect with any translations of itself. We will do this for $T_j$ and $p'_j$ separately.
		\begin{enumerate}
			\item $T_j \bigcap \left(T_j +\sum_{i=1}^{j+1} k_i\ket{w_i}\right)$ for $k \neq 0^{j+1}$\\
			If $k_{j+1}=0$, then this intersection is empty by the inductive hypothesis.
		If $k_{j+1}\neq 0$, then this intersection is empty because $\ket{w_{j+1}}= l_{j+1}\ket{v_{j+1}}+\sum_{i=1}^{j} l_i \ket{w_{i}}$ and $l_{j+1}>s_j$, the largest integer such that $S_j \cap S_j+s_j\ket{v_{j+1}}$ is non-empty.
		
		\item $p'_j \bigcap \left(T_j +\sum_{i=1}^{j+1} k_i\ket{w_i}\right)$ for $k \neq 0^{j+1}$\\
		The path $p_j$ will not intersect with any translation of $T_j$ because it is a shortest path between the surfaces $S_j$ and $S_j+\ket{w_{j+1}}$.
		\item $p'_j \bigcap \left(p'_j +\sum_{i=1}^{j+1} k_i\ket{w_i}\right)$ for $k \neq 0^{j+1}$\\	By the definition of $p'_j$, this intersection is empty for $k_{j+1}>0$. And it is also empty for $k_{j+1}<0$, by a symmetric argument. 
		
		So now suppose for a contradiction that $p_j$ and  $p_j+\sum_{i=1}^j k_i \ket{w_i}$ intersect at a vertex $y$. Then since both paths are shortest paths between $S_j$ and $S_j+\ket{w_{j+1}}$, $y$ must be at the same position along the two paths (equal to the distance from $S_j$). But since one path is a translation of the other, we must have $y=y+\sum_{i=1}^j k_i \ket{w_i}$, which is a contradiction.
		
		%Now suppose $p_j$ intersects with a translation of $p_j$ that goes from $S_j+m\ket{w_{j+1}}$ and $S_j+(m+1)\ket{w_{j+1}}$ for some $m>1$.
		%This gives four paths to the central vertex $x$, all of which must be the same length, in order to satisfy the shortest path condition, implying that they meet at the midpoints of their length. But therefore this midpoint should be at its location as well as its translation. Contradiction!
		\end{enumerate}

This inductive proof gives a deterministic iterative process by which one can construct a subgraph $T$ and the basis $\{\ket{w_i}\}_{i=1}^D$.	
	
\end{proof}

We are now ready to prove the main result of this section:

\begin{reptheorem}{thm:TIsimGL}
	Let $G$ be a connected translationally-invariant graph embedded in $\R^D$ for $D\ge 2$.
	Let $\cS$ be $\{XX+YY+ZZ\}$ or $\{XX+YY\}$.
	Then $\cS_G$-Hamiltonians are universal and can efficiently simulate all geometrically local Hamiltonians in $\R^D$.
\end{reptheorem}

\begin{proof}
	Recall that by Lemma~\ref{lem:GL=cubic}, $\cS$-Hamiltonians can efficiently simulate all geometrically local Hamiltonians in $\R^D$, even when restricted to the edges of a degree 3 subgraph of the hypercubic lattice in $\R^D$.
	And that by \Cref{thm:XYsquare}, $\cS$-Hamiltonians are universal even when restricted to the edges of a 2D square lattice, which is a subgraph of the hypercubic lattice in $\R^D$ for all $D\ge 2$.
	
	It therefore remains to show that $\cS_G$-Hamiltonians can efficiently simulate any $\cS$-Hamiltonian $H$ whose interactions are restricted to a degree 3 subgraph $F$ of the hypercubic lattice.
	We will show that one can subdivide the edges of $F$ to obtain a subgraph of $G$. With the subdivision gadgets of \Cref{lem:PMgadgets}, this will suffice to prove the claim.
	
	By Lemma~\ref{lem:TIgraph=cubic} there exists a subgraph $T\subseteq G$ and a basis $\{\ket{w_i}\}_{i=1}^D$ such that for all $k \in \Z^D$, the subgraphs $T(k)=T+\sum_{i=1}^D k_i \ket{w_i}$ are disjoint and there is an edge between $T(k)$ and $T(m)$ in $G$ if there is an edge between $k$ and $m$ in the hypercubic lattice. The subgraph $T$ and the basis $\{\ket{w_i}\}_{i=1}^D$ can be found using the constructive proof of Lemma~\ref{lem:TIgraph=cubic} in time which depends only on $G$ (and not on the target Hamiltonian $H$ or its interaction graph $F$).
	
	%For each vertex $k \in \Z^D$ in the hypercubic lattice, we will identify $k$ with a vertex in $T(k)$, and treat all other qubits in $G$ as ancilla qubits for the subdivision gadgets.
	Since $F$ is a degree 3 subgraph of the hypercubic lattice, $k$ shares an edge with at most three other vertices $k(1),k(2),k(3)$ in $F$. 
	Let $x_1,x_2,x_3$ be the three vertices in $T(k)$ which have an edge with $T(k(1)),T(k(2)),T(k(3))$ respectively. 
	%These vertices exist by Lemma~\ref{lem:TIgraph=cubic}.
	We will identify each vertex $k \in F$ with the central point $y(k)$ in $T(k)$ such that there are three non-overlapping paths from $y(k)$ to each $x_i$. 
	To see that such a central vertex $y$ exists for any three vertices $x_1,x_2,x_3$ in a finite connected graph $T$, first consider a path from $x_1$ to $x_2$ and then consider the shortest path from $x_3$ to this path. The point where these two paths meet is such a central vertex.
	
	We have now identified each vertex $k \in F$ with a vertex $y(k) \in G$ such that for each edge $(k,m) \in F$ there is a corresponding path form $y(k)$ to $y(m)$ in $G$, and that none of these paths overlap.
	Using the subdivision gadgets of \Cref{lem:PMgadgets}, we can simulate an $h$ interaction along each of these paths, to simulate the whole Hamiltonian $H$.
	
	Since $T$ contains only finitely many vertices (independent of $n$, the number of qubits of the target Hamiltonian $H$), finding the central vertex for each translation of $T$ takes $\BigO(1)$ time, and so finding all such central vertices takes time $\BigO(n)$. 
	Furthermore, each edge in $F$ is subdivided at most $\BigO(1)$ times and so only $\BigO(1)$ rounds of perturbation theory are required and thus the simulation is efficient.
\end{proof}
%\begin{figure}
%	\centering
%	\begin{tikzpicture}
%	\newcommand{\vone}{1.2,0}
%	\newcommand{\vtwo}{0.2,1.5}
%	
%	\pgfmathsetmacro{\x}{4}
%	\pgfmathsetmacro{\y}{4}
%	\foreach \j  in {0,...,\y}{
%		\foreach \i in {0,...,\x}{
%			\begin{scope}[shift={($\i*(\vone)+\j*(\vtwo)$)}]
%			\node[qubit] (u0) at (0.1,0.7){};
%			\node[qubit] (u1) at (0,0) {};
%			\node[qubit] (u2) at (0.6,-0.4) {};
%			\node[qubit] (u4) at (-0.6,0.4){};
%			\draw (u1) -- ($(u2)+(\vone)$);
%			\draw (u1) -- (u2);
%			\draw (u0) -- (u1);
%			\draw (u4) -- (u1);
%			\draw (u4) -- ($(u4)-(\vtwo)+(\vone)$);
%			\end{scope}}}
%		
%		\begin{scope}[blue,shift={($2*(\vone)+3*(\vtwo)$)}]
%		\node[qubit] (u0) at (0.3,0.7){};
%		\node[qubit] (u1) at (0,0) {};
%		\node[qubit] (u2) at (0.7,-0.4) {};
%		\node[qubit] (u4) at (-0.4,0.6){};
%		\draw (u1) -- ($(u2)+(\vone)$);
%		\draw (u1) -- (u2);
%		\draw (u0) -- (u1);
%		\draw (u4) -- (u1);
%		\draw (u1) -- ($(u4)-(\vtwo)+(\vone)$);
%		\end{scope}
%	
%
%	
%	\end{tikzpicture}
%\end{figure}

\section{Discussion}\label{sec:conclusion}
In this work we have shown that translationally-invariant interactions suffice to define universal Hamiltonian simulators.
Naturally---due to parameter counting---the size of the simulated target Hamiltonian is at most polylogarithmic in the size of the simulator; or in other words, there is an at least exponential spacial overhead present when simulating a non-translationally-invariant Hamiltonian with a translationally-invariant one.
Indeed, if the target Hamiltonian itself features translational symmetry, the overhead is at most polynomial, as shown in \Cref{sec:TItarget}.

In addition to the parameter counting overhead we also need to scale $\Lambda\Sim$ in accordance with the simulation precision; for fixed accuracy, the overhead just mentioned suffices. Yet since we need to be able to simulate a target system to \emph{arbitrary} precision, the overhead of the simulating Hamiltonian's system size could be accordingly larger.

While our construction demonstrates that many-body localization and similar effects (in fact, \emph{any} many-body effect) present in lattice systems without translational symmetry can be simulated in a translationally-invariant fashion, we again emphasize that due to the exponential space overhead the sector of the global Hilbert space in which these effects occur is itself doubly-exponentially small.
This effect is often generic in models with many-body localization \cite{Pietracaprina2019}; yet in our case it is an effect of the simulator's overhead, not intrinsic to the physics describing the target Hamiltonian.
%TODO maybe mention quantum scars; 

Instead of simulating a specific non-translationally-invariant Hamiltonian, one can of course obtain a source of (pseudo-) randomness in a simpler fashion:
just like how the history state Hamiltonian in \Cref{sec:-history} is used to extract a background field of real numbers $\{(a_i,b_i)\}_{i\in\Lambda\Sq}$---which condition the coupling strengths of the target Hamiltonian---we can use the universal computational capabilities of the encoded Turing machine to produce a field of pseudo-random numbers from a small initial seed (e.g.\ obtained from the system size, and calculated by a Mersenne Twister).
Like that, the entire lattice area of $\Lambda\sim$ can be used to reproduce effects of many-body localization.
As it seems arguable that true randomness is required to observe many-body localization effects, this would result in a translationally-invariant lattice which, within its low-energy sector, reproduces a local Hamiltonian with randomly-varying coupling strengths across the spin lattice, and thus all observable effects emerging from this setup.

There are several open questions.
The first one is whether a similar universality result holds for translationally-invariant spin chains.
Since perturbation gadgets do not work in one dimension, we need to resort to different proof techniques. Another question left open is that of optimizing the polynomial/exponential overhead necessary in the number of qudits, as well as optimizing the local dimension; this question is of particular importance in case a simulation reduction is to be used in an actual experimental or numerical study setup.
Secondly, apart from translational invariance, there are of course other symmetries---such as mirror symmetry, or discrete rotational symmetries.
It is known that the local Hamiltonian problem for translationally-invariant spin chains remains QMA\textsubscript{EXP}-hard in the mirror-symmetric setting \cite[Ch.~6]{Gottesman2009}; yet how a universal Hamiltonian can emerge from this setup is not immediately obvious.
Finally it would of course be desirable to reduce the local dimension to be as low as possible while still retaining universality.

\section*{Acknowledgements}
We are grateful for helpful discussions with Nigel Cooper, Andrew Green and Robin Reuvers regarding many-body localization.
J.\,B.\ acknowledges support from the Draper's Research Fellowship at Pembroke College.
No new data were created during this study.

\printbibliography

\appendix
\end{document}

%% file: tiles.tex
\newcommand{\dx}{0.02}
\newcommand{\ts}{.95cm}
\newcounter{tx}
\newcounter{ty}

\usepackage{xparse}
\usepackage{tikz}
\usetikzlibrary{patterns, decorations.pathreplacing}

\makeatletter
\tikzset{
    hatch distance/.store in=\hatchdistance,
    hatch distance=4pt,
    hatch thickness/.store in=\hatchthickness,
    hatch thickness=2pt
}
\pgfdeclarepatternformonly[\hatchdistance,\hatchthickness]{stripey}
    {\pgfqpoint{-1pt}{-1pt}}% below left
    {\pgfqpoint{\hatchthickness}{10pt}}% above right
    {\pgfqpoint{\hatchdistance}{10pt}}% offset to next
    {
        \pgfsetcolor{\tikz@pattern@color}
        \pgfsetlinewidth{\hatchthickness}
        \pgfpathmoveto{\pgfqpoint{0pt}{-1pt}}
        \pgfpathlineto{\pgfqpoint{0pt}{10pt}}
        \pgfusepath{stroke}
    }
\makeatother

\DeclareDocumentEnvironment{tiles}{ s O{1} }{%
  \setcounter{tx}{0}%
  \setcounter{ty}{0}%
  \IfBooleanTF{#1}{%
    \begin{tikzpicture}[
      draw=black,
      font=\scriptsize,
      line width=\dx*\ts,
      x=#2*\ts,y=#2*\ts,
      baseline={([yshift=-.6ex]current bounding box.center)}
    ]
  }{%
    \begin{tikzpicture}[
    draw=black,
    font=\scriptsize,
    line width=\dx*\ts,
    x=#2*\ts,y=#2*\ts
    ]
  }%
}{
  \end{tikzpicture}%
}

\renewcommand\${
\setcounter{tx}{0}
\addtocounter{ty}{-1}
}

\DeclareDocumentCommand{\tile}{m o}{
  \addtocounter{tx}{1}
  \addtocounter{ty}{0}
  \begin{scope}[
    shift={(\thetx,\thety)}
  ]
  #1
  \draw[fill=white] (.1,.1) rectangle (.9,.9);
  \IfValueT{#2}{
  		\node[] at (.5,.5) {#2};
  }
  \end{scope}
}

\DeclareDocumentCommand{\defineEdge}{ m m m m m }{
	\DeclareDocumentCommand{#1}{ s m } {
		\draw[fill=colour##2,colour##2,rounded corners=1] (#2) rectangle (#3);
		\IfBooleanT{##1}{\draw[fill=white,white,line width=2*\dx*\ts] (#4) rectangle (#5);}
	}
}
\defineEdge\edgeN{.3,.8}{.7,1}{.45,.8}{.55,1}
\defineEdge\edgeE{.8,.3}{1,.7}{.8,.45}{1,.55}
\defineEdge\edgeS{.3,0}{.7,.2}{.45,0}{.55,.2}
\defineEdge\edgeW{0,.3}{.2,.7}{0,.45}{.2,.55}

\DeclareDocumentCommand{\Tomg}{m m m m o}{%
	\tile{\edgeN#1\edgeE#2\edgeS#3\edgeW#4}[#5]
}
\def\T(#1,#2,#3,#4) {\Tomg{#1}{#2}{#3}{#4}}
\def\TT(#1,#2,#3,#4,#5) {\Tomg{#1}{#2}{#3}{#4}[#5]}

\def\Ts(#1,#2,#3,#4){\begin{tiles}*%
\Tomg{#1}{#2}{#3}{#4}%
\end{tiles}}

\def\TTs(#1,#2,#3,#4,#5){\begin{tiles}*%
\Tomg{#1}{#2}{#3}{#4}[#5]%
\end{tiles}}

\definecolor{colour-}{RGB}{225,225,225}
\definecolor{colour0}{RGB}{0,0,0}
\definecolor{colour1}{RGB}{155,155,155}
\definecolor{colourR}{RGB}{255,34,44}
\definecolor{colourB}{RGB}{11,92,255}
\definecolor{colourD}{RGB}{200,0,100}

\definecolor{colourT}{RGB}{192,192,192}
\definecolor{colourS}{RGB}{235,235,235}

\newcommand\TMstate[2]{\shortstack[c]{$#1$\\$#2$}}